\documentclass[10pt]{IEEEtran}  
\usepackage{graphicx, xcolor, cite}
\usepackage{amsmath, amsthm, amsfonts, enumerate, amssymb, setspace, booktabs, threeparttable}
\usepackage{acronym}
\usepackage{algorithm}
\usepackage[noend]{algpseudocode}
\usepackage{balance}

\usepackage[fancythm,fancybb]{jphmacros} 

\usepackage{algorithm}
\usepackage{algpseudocode}
\newcommand\ignore[1]{{}}

\DeclareMathOperator*{\arginf}{arg\,inf}

\theoremstyle{definition}
\usepackage{mathtools}

\usepackage{filecontents}

\newcommand{\lip}{\mathfrak{L}^*_\phi}
\newcommand{\liphat}{\hat{\mathfrak{L}}_\phi}

\usepackage{graphicx, xcolor}
\usepackage{multirow}
\usepackage{color, colortbl}
\definecolor{RoyalBlue}{rgb}{0.9,1,1}

\acrodef{kde}[KDE]{kernel density estimation}
\acrodef{lmis}[LMIs]{linear matrix inequalities}
\acrodef{ard}[ARD]{automatic relevance determination}
\acrodef{adp}[ADP]{adaptive/approximate dynamic programming}

\title{Safe Approximate Dynamic Programming Via Kernelized Lipschitz Estimation}
\author{Ankush Chakrabarty$^{1,\dag}$, Devesh K. Jha$^1$, Gregery T. Buzzard$^2$, Yebin Wang$^1$, Kyriakos G. Vamvoudakis$^3$
	\thanks{$^1$Mitsubishi Electric Research Laboratories, Cambridge, MA, USA. Email: \texttt{\{chakrabarty,devesh.jha,yebinwang\}@merl.com}}%
	\thanks{$^2$Department of Mathematics, Purdue University, West Lafayette, IN, USA. Email: \texttt{buzzard@purdue.edu}}%
	\thanks{$^3$Daniel Guggenheim School of Aerospace Engineering, Georgia Institute of Technology, Atlanta, GA, USA. Email: \texttt{kyriakos@gatech.edu}}%
	\thanks{$^\dag$Corresponding author: A.~Chakrabarty. Phone: +1~(617)~758-6175.}}

\begin{document}
%
%
		
	\maketitle
	\begin{abstract}
	We develop a method for obtaining safe initial policies for reinforcement learning via approximate dynamic programming (ADP) techniques for uncertain systems evolving with discrete-time dynamics. We employ kernelized Lipschitz estimation and semidefinite programming for computing admissible initial control policies with provably high probability. Such admissible controllers enable safe initialization and constraint enforcement while providing exponential stability of the equilibrium of the closed-loop system.
	\end{abstract}
	\begin{IEEEkeywords}
	Semidefinite programming; Lipschitz constant estimation; linear matrix inequalities; neural networks; policy iteration; value iteration; kernel density estimation; approximate dynamic programming; incremental quadratic constraints.
	\end{IEEEkeywords}

\section{Introduction}
Recent advances in the field of deep and machine learning has led to a renewed interest in using learning for control of physical systems~\cite{vamtutorial}. 
Reinforcement learning (RL) is a learning framework that handles sequential decision-making problems, wherein an `agent' or decision maker learns a policy to optimize a long-term reward by interacting with the (unknown) environment. At each step, an RL agent obtains evaluative feedback (called reward or cost) about the performance of its action, allowing it to improve the performance of subsequent actions \cite{sutton1998reinforcement,vrabie2013optimal}. While RL has witnessed huge success in recent times~\cite{silver2016mastering,silver2017mastering}, there are several unsolved challenges which restricts use of these algorithms for industrial systems. In most practical applications, control policies must be designed to satisfy operational constraints. This leads to the challenge that one has to guarantee constraint satisfaction during learning and policy optimization. Therefore, initializing with an unverified control policy is not `safe' (in terms of stability or constraint handling). In other words, using on-line RL for expensive equipment or safety-critical applications necessitates that the initial policy used for obtaining data for subsequently improved policies must be at least stabilizing, and generally, constraint-enforcing. The work presented in this paper is motivated by this challenge. We  present a framework for deriving initial control policies from historical data  that can be verified to satisfy constraints and guarantee stability while learning the optimal control policy on-line, from operational data. 

A successful RL method needs to balance a fundamental trade-off between exploration and exploitation. One needs to gather data safely (exploration) in order to best extract information from this data for optimal decision-making (exploitation). One way to solve the exploration and exploitation dilemma is to use optimistic initialization \cite{7329650,brafman2002r,thomas2015high,jha2016data}, but this assumes the optimal policy is available until data is obtained that proves otherwise. Such approaches have been applied to robotics applications, where systems with discrete and continuous state-action spaces~\cite{levine2016end,duan2016benchmarking}.
A limitation of these methods is that, before the optimal policy is learned, the agent is quite likely to explore actions that lead to violation of the task-specific constraints as it aims to optimize the cumulative reward for the task. This shortcoming significantly limits such methods to be applicable to industrial applications, since this could lead to irreparable hardware damage or harm human operators due to unexpected dynamics. Consequently, \textit{safe learning} focuses on learning while enforcing safety constraints. 
There are primarily two types of approaches to safe RL and approximate/adaptive dynamic programming (ADP). These include: modification of the optimization criterion with a safety component such as barrier functions by transforming the operational constraints into soft constraints~\cite{achiam2017constrained, chow2018lyapunov}; and, modifying the exploration process through the incorporation of external system knowledge or historical data~\cite{garcia2015comprehensive}. Our method is amongst the latter class of methods, because our operational constraints are hard constraints and softening them could lead to intermittent failure modes. 

High performance model-based control requires precise model knowledge for controller design. However, it is well known that for most applications, accurate model knowledge is practically elusive due to the presence of unmodeled dynamical interactions (e.g., friction, contacts, etc.). 
Recent efforts tackle this issue by learning control policies from operational (on-line) or archival data (off-line).
Since the exact structure of the nonlinearity may be unknown or not amenable for analysis, researchers have proposed `indirect' data-driven controllers that employ non-parametric learning methods such as Gaussian processes to construct models from operational data~\cite{romeres2016, DBLP:journals/corr/abs-1809-04993} to improve control policies on-line~\cite{Berkenkamp2016,hewing2017}.	
Conversely, `direct' methods, such as those proposed in~\cite{jiang2015optimal,tanaskovic2017data,piga2018direct,kiumarsi2018optimal}, directly compute policies using a combination of archival/legacy and operational input-output data without constructing an intermediate model. For example, in~\cite{fagiano2014automatic}, a human expert was introduced into the control loop to conduct initial experiments to ensure safety while generating archival data. A common assumption in many of these approaches is the availability of an initial control policy that is stabilizing and robust to unmodeled dynamics. Designing such \textit{safe} initial control policies in a computationally tractable manner remains an open challenge.

In this work, we present a formalism for synthesizing safe initial policies for uncertain non-linear systems. We assume the presence of historical/archival/legacy data, with which we estimate Lipschitz constants for the unmodeled system dynamics. The estimation of the Lipschitz constant is done via kernel density estimation (KDE). The estimated Lipschitz constant is used to design control policies via semidefinite programming that can incorporate stability and constraint satisfaction while searching for policies. We show that the proposed approach is able to design feasible policies for different constrained tasks for several  systems while respecting all active constraints. 
Our key insight is that information regarding the structure of classes of unmodeled nonlinearities can be encapsulated using only a few parameters, without knowing the exact form of the nonlinearity. Therefore, it may not be necessary to model the unknown component itself in order to compute a safe control policy. For example, the class of Lipschitz nonlinearities (which constitute a large share of nonlinearities observed in applications) can be described using only a few parameters: the Lipschitz constants of the nonlinear components. Recent work has investigated the utility of Lipschitz properties in constructing controllers when an oracle is available~\cite{chakrabarty2017support} or in designing models for prediction~\cite{calliess2014conservative} with on-line data used for controller refinement~\cite{limon2017learning}
In this paper, we construct control policies that respect constraints and certify stability (with high probability) for applications \textit{where only off-line data is available, and no oracle is present}. We do so through the systematic use of multiplier matrices that enable the representation of nonlinear dynamics through quadratic constraints~\cite{chakrabarty2017state,xu2018observer} without requiring knowledge of the underlying nonlinearity. The control policies can then be obtained by solving semidefinite programs.
However, construction of multiplier matrices for Lipschitz systems requires knowledge of the Lipschitz constants, which are not always available, and therefore, must be estimated. We refer to the estimation of Lipschitz constants from data as~\textit{Lipschitz learning}. 
%
Historically, methods that estimate the Lipschitz constant~\cite{wood1996estimation,strongin1973convergence,hansen1992using} do not provide certificates on the quality of the estimate. Herein, we provide conditions that, if satisfied, enable us to estimate the Lipschitz constant of an unknown locally Lipschitz nonlinearity with high probability. To this end, we employ kernel density estimation (KDE), a non-parametric data-driven method that employs kernels to approximate smooth probability density functions to arbitrarily high accuracy. We refer to our proposed KDE-based Lipschitz constant estimation algorithm as~\textit{kernelized Lipschitz learning}. 

\subsubsection*{Contributions}
Compared to the existing literature on safe learning, the contributions of the present paper are threefold. First, we formulate an algorithm to construct stabilizing and constraint satisfying policies for nonlinear systems without knowing the exact form of the nonlinearity. Then we leverage a kernelized Lipschitz learning mechanism to estimate  Lipschitz constants of the unmodeled dynamics with high probability; and, finally we use a multiplier-matrix based controller design based on Lipschitz learning from legacy data that forces exponential stability  on the closed-loop dynamics (with the same probability as the kernelized Lipschitz learner).

\paragraph*{Structure} The rest of the paper is structured as follows.  We present the formal motivation of our work in Section~\ref{sec:prelim}. Our kernelized Lipschitz learning algorithm is described in Section~\ref{sec:kde}, and benchmarking of the proposed learner on benchmark Lipschitz functions is performed. The utility of Lipschitz learning in policy design via multiplier matrices is elucidated in Section~\ref{sec:control}, and a numerical example demonstrating the potential of our overall formalism is provided in Section~\ref{sec:ex}. We provide concluding remarks and discuss future directions in Section~\ref{sec:conc}.
\paragraph*{Notation}
We denote by $\mathbb{R}$ the set of real numbers, $\mathbb R_+$ as the set of positive reals, and $\mathbb{N}$ as the set of natural numbers. The measure-based distance between two measurable subsets $A$ and $B$ of a metric space $\mathbb R^n$ equipped with the metric $\rho_\mu$ is given by
$\rho_\mu(A, B) = \mu(A\triangle B)$, where $\mu$ is a measure on $\mathbb R^n$ and $A\triangle B$ is the symmetric difference $(A\setminus B)\cup (B\setminus A)$.
We define a ball $\mathcal B_{\epsilon}(x) := \{y:\rho(x,y)\le \epsilon\}$ and the sum $A\oplus \epsilon:=\bigcup_{x\in A} \mathcal B_{\epsilon}(x)$. The complement of a set $A$ is denoted by $A^c$. The indicator function of the set $A$ is denoted by $\mathbf 1_A$. A block diagonal matrix is denoted by $\mathrm{blkdiag}\big(\cdot\big)$.
For every $v\in\mathbb{R}^n$, we denote $\|v\|=\sqrt{v^\top v}$, where $v^\top$ is the transpose of $v$. The sup-norm or $\infty$-norm is defined as $\|v\|_\infty \triangleq \sup_{t\in\mathbb{R}}\|v(t)\|$. We denote by $\lambda_{\min}(P)$ and $\lambda_{\max}(P)$ as the smallest and largest eigenvalue of a square, symmetric matrix $P$. The symbol $\succ(\prec)$ indicates positive (negative) definiteness and $A\succ B$ implies $A-B\succ 0$ for $A,B$ of appropriate dimensions. Similarly, $\succeq (\preceq)$ implies positive (negative) semi-definiteness. The operator norm is denoted $\|P\|$ and is defined as the maximum singular value of $P$. For a symmetric matrix, we use the $\star$ notation to imply symmetric terms, that is,
$
\left[\begin{smallmatrix} a & b \\ b^\top & c\end{smallmatrix}\right] \equiv \left[\begin{smallmatrix}
a & b \\ \star & c
\end{smallmatrix}\right]$. The symbol $\mathbf{Pr}$ denotes the probability measure.

	\section{Problem Formulation}\label{sec:prelim}
	
	\subsection{Problem statement}
	Consider the following discrete-time nonlinear system,
	\begin{align*}
	x_{t+1} &= F(x_t,u_t),\ t\in\mathbb N\\
	q_t &= C_q x_t,
	\end{align*}
	where $x_t\in\mathbb{R}^{n_x},\ u=u_t\in \mathbb{R}^{n_u}$ denote the state and the control input of the system respectively.

	For simplicity of exposition we will write 
	\begin{subequations}
		\label{eq:true_sys}
		\begin{align}
		x_{t+1} &= Ax_t + Bu_t + G\phi(q_t),\ t\in\mathbb N\\
		q_t &= C_q x_t,
		\end{align}
	\end{subequations}
	where the system matrices $A$, $B$, $G$ and $C_q$ have appropriate dimensions. Denote by $\phi\in\mathbb{R}^{n_\phi}$ the system's uncertainty, or unmodeled nonlinearity, whose argument $q=q_t\in \mathbb{R}^{n_q}$ is represented by a linear combination of the state. The origin is an equilibrium state for~\eqref{eq:true_sys}; that is, $\phi(0)=0$.
	
	The following assumptions and definition are now needed.
	\begin{assumption}\label{asmp:matrices_knowledge}
	The matrix $B$ is known. The matrix $G$ has full column rank and is sparse and all entries are 0 or 1. Only the non-zero element locations of $G$ are known. The matrix $C_q$ and function $\phi$ are unknown.\frqed
	\end{assumption}
	
	We require the following definition to describe the class of nonlinearities considered in this paper.
	\begin{definition}
		A function $f:\mathbb X\to \mathbb R^{n_x}$ is Lipschitz continuous in the domain $\mathbb X\subset \mathbb R^{n_f}$ if
		\begin{equation}\label{eq:lipsz}
			\|f(x_1) - f(x_2)\|\le \mathfrak L_f\|x_1 - x_2\|
		\end{equation}
		for some $\mathfrak L_f>0$ and all $x_1, x_2\in \mathbb X$. We define the scalar
		\begin{equation}\label{eq:best_lips}
			\mathfrak L_f^* = \inf_{\mathbb R_{+}} \{\mathfrak L_f: \text{condition~\eqref{eq:lipsz} holds}\}
		\end{equation}
		as the Lipschitz constant of $f$ in $\mathbb X$. A function is globally Lipschitz if~\eqref{eq:lipsz} holds for $\mathbb X\equiv\mathbb R^{n_f}$. \frqed
	\end{definition}
	\begin{assumption}\label{asmp:nonlin_knowledge}
		The nonlinearity $\phi$ is globally Lipschitz continuous. That is,
		\begin{equation}
			\label{eq:PhiLipschitz}
			\|\phi(q_1) - \phi(q_2)\|\le \lip \|q_1 - q_2\|
		\end{equation}
		for any $q_1, q_2\in\mathbb R^{n_q}$, and the global Lipschitz constant $\lip$ is unknown.\frqed
	\end{assumption}
	Assumptions~\ref{asmp:matrices_knowledge} and~\ref{asmp:nonlin_knowledge} imply that the linear component of the true system~\eqref{eq:true_sys} can be assumed, but the rest is unknown. However, we do know the vector space through which the nonlinearity enters the dynamics of~\eqref{eq:true_sys}, since the non-zero locations of $G$ are flagged.

	\begin{remark}
	Assumption 1 is mild. For instance, one could relax the assumption on $G$ and take the unknown $\tilde G$ to be the identity matrix. 
	Then the nonlinearity would be $\tilde\phi(q) = \begin{bmatrix} 0 & \ldots & \phi_i(q) & \ldots \end{bmatrix}^\top,\ i\in \mathcal I$, with $\mathcal I$ the index set of non-zero rows of $G$, so that $\tilde G\tilde \phi=G\phi$. \frqed
	\end{remark}

	Given a control policy $u(x)$, we define an infinite horizon cost functional given an initial state $x_0\in\mathbb R^{n_x}$ as
	\begin{equation}
	\label{eq:cost}
	\mathcal J(x_0, u) = \sum_{t=0}^\infty \gamma^t\,\mathcal U(x_t, u(x_t)),
	\end{equation}
	where $\mathcal U$ is a function with non-negative range, $\mathcal U(0,0)=0$, and $\{x_k\}$ denotes the sequence of states generated by the closed loop system
	\begin{equation}\label{eq:cl_sys}
	x_{t+1} = Ax_t + Bu(x_t) + G\phi\left(C_q x_t \right).
	\end{equation}
	The scalar $\gamma\in (0,1]$ is a forgetting/discount factor intended to enable the cost to be emphasized more by current state and control actions and lend less credence to the past.
	
	Before formally stating our objective, we need to introduce the following standard definition~\cite{vamtutorial}.
	\begin{definition}
		\label{defn:admissible_policy}
		A continuous control policy $u(\cdot):\mathbb R^{n_x}\to \mathbb R^{n_u}$ is \textit{admissible} on $X\subset \mathbb R^{n_x}$ if it stabilizes the closed loop system~\eqref{eq:cl_sys} on $X$
		and	$\mathcal J(x_0,  u)$ is finite for any $x_0\in X$.
	\end{definition}
	We want to design an optimal control policy that achieves the optimal cost
	\begin{equation}\label{eq:J_opt}
	\mathcal J_\infty (x_0) = \inf_{u\in\mathfrak U} \mathcal J\big(x_0,u\big),
	\end{equation}
	for any $x_0\in\mathbb R^{n_x}$. Here, $\mathfrak U$ denotes the set of all admissible control policies. In other words, we wish to compute an optimal control policy
	\begin{equation}\label{eq:u_opt}	
	u_\infty = \arginf_{u\in\mathfrak U} \mathcal J\big(x_0,u\big).
	\end{equation}
	Directly constructing such an optimal controller is very challenging for general nonlinear systems; this is further complicated because the system~\eqref{eq:true_sys} contains unmodeled/uncertain dynamics. Therefore, we shall use adaptive/approximate dynamic programming (ADP): a class of iterative, data-driven algorithms that generate a convergent sequence of control policies whose limit is provably the optimal control policy $u_\infty(x)$. 
	
	Recall from~\cite{lewis2012reinforcement,7823092} that a necessary condition for convergence of policy iteration methods (a sub-class of ADP) is the availability of an initial admissible control policy $u_0(x)$, which is non-trivial to derive for systems with some unmodeled dynamics. Therefore, our objective in this work is to systematically derive an initial admissible control policy using only partial model information via kernelized Lipschitz learning and semidefinite programming. We also extend this idea to handle the case when the control input is constrained. In such cases, along with an admissible controller, we also derive a domain of attraction of the controller within which the control policy is guaranteed to satisfy input constraints and the closed-loop system remains stable. We refer to the derivation of admissible control policies with guaranteed stabilizability and/or constraint enforcement as \textit{safe initialization for ADP}: a crucial property required for ADP algorithms to gain traction in expensive industrial applications.
	
	We invoke the assumption in~\cite{tanaskovic2017data,piga2018direct} regarding the availability of legacy/archival/historical data generated by the system during prior experiments. That is, at design time, we have a  dataset $\mathcal D$ consisting of unique triples: state-input pairs along with corresponding state update information. Concretely, we have access to $\mathcal D = \{x_j, u_j, x^+_j\}_{j=1}^{N}$.
	For each $\{x_j, u_j, x^+_j\}\in\mathcal D$, we estimate the nonlinear term using~\eqref{eq:true_sys}; that is,
	\[
	\phi(q_j) = G^\dag\left(x^{+}_{j} - Ax_j - Bu_j\right),
	\]
	where $G^\dag$ exists by Assumption~\ref{asmp:matrices_knowledge}. Note that we also need to estimate the matrix $C_q$ (see~\eqref{eq:true_sys}) so that $q_j$ can be calculated from $x_j$. While estimating the exact elements of these matrices is quite challenging, we can estimate the non-zero elements in the matrices, which is enough to design safe initial control policies, because the exact elements of $C_q$ will be subsumed within the Lipschitz constant. 
	
	\begin{remark}
	The problem of estimating the sparsity pattern of $C_q$ is analogous to the problem of feature selection and sparse learning, known as automatic relevance determination (ARD)~\cite{tipping2001sparse}. The basic idea in ARD is to give feature weights some parametric prior densities;  these densities are subsequently refined by maximizing the likelihood of the data~\cite{tipping2001sparse, chu2005preference}. For example, one can define hyperparameters which explicitly represent the relevance of different inputs to a machine learning algorithm w.r.t. the desired output (e.g., a regression problem). These relevance hyperparameters determine the range of variation of parameters relating to a particular input. ARD can then determine these hyperparameters during learning to discover which inputs are relevant.\frqed\end{remark}
	
	We need the following assumption on the data $\{q_j, \phi(q_j)\}$, without which one cannot attain the global Lipschitz constant of the nonlinearity $\phi(\cdot)$ with high accuracy.
	\begin{assumption}
	\label{asmp:quality_of_q_samples}
	Let $\mathcal Q$ denote the convex hull of the samples $\{q_j\}$. The Lipschitz constant of $\phi(\cdot)$ in the domain $\mathcal Q$ is identical to the global Lipschitz constant $\lip$.\frqed
	\end{assumption} 
	Assumption~\ref{asmp:quality_of_q_samples} ensures that the samples obtained from the archival data are contained in a subregion of $\mathbb R^{n_q}$ where the nonlinearity $\phi(\cdot)$'s local Lipschitz constant is the same as its global Lipschitz constant.
	
	\begin{example}\label{ex:2}
	Suppose $\phi(q)=1.5\sin(q)$. As long as the convex hull of the samples $\{q\}$ contain zero, the Lipschitz constant of $\phi$ on the convex hull $\mathcal Q$ and on $\mathbb R$ are identical.	\frqed
	\end{example}
	
	In the following section, we will leverage the dataset $\mathcal D$ to estimate the Lipschitz constant of $\phi(\cdot)$ using kernelized Lipschitz learning/estimation, and consequently design an initial admissible linear control policy
	$u_0 = K_0 x$ via semidefinite programs. We will demonstrate how such an initial admissible linear control policy fits into a neural-network based ADP formulation (such as policy iteration) to asymptotically generate the optimal control policy $u_\infty(x)$.
	\begin{remark}
	The control algorithm proposed in this paper is a direct data-driven controller because no model of $\phi(\cdot)$ is identified in the controller design step.\frqed
	\end{remark}
	\begin{remark}
	Although we focus only on discrete-time systems, our results hold for continuous-time systems with slight modifications.\frqed
	\end{remark}
	\begin{remark}
	If $n_\phi>1$, our proposed Lipschitz learning algorithm will yield $n_\phi$ Lipschitz constant estimates, one for each dimension of $\phi(\cdot)$. To avoid notational complications, we proceed (without loss of generality) with $n_\phi = 1$. For larger $n_\phi$, our algorithm can be used component-wise.\frqed
	\end{remark}
	
	\section{Kernelized Lipschitz Learning}\label{sec:kde}
	In this section, we provide a brief overview of kernel density estimation (KDE) and provide a methodology for estimating Lipschitz constants from data. 
    
    \subsection{Empirical density of Lipschitz estimates}	
	With the data $\{\phi(q_j), q_j\}_{j=1}^N$, we obtain $n\in\mathbb N$  underestimates of the global Lipschitz constant $\lip$ using
	\begin{equation}\label{eq:lipschitz_underestimates}
		\ell_{jk} = \frac{|\phi(q_j) - \phi(q_k)|}{\|q_j - q_k\|},
	\end{equation}
	where $k\in \{1,\ldots, N\}\setminus j$. The sequence $\{\ell_{jk}\}$ are empirical samples drawn from an underlying univariate distribution $L$. Clearly, the true distribution $L$ has finite support; indeed, its left-hand endpoint is a non-negative scalar (zero, if $n_q > 1$ but may be positive if $n_q = 1$) and its right-hand endpoint is $\lip$. This leads us to the key idea of our approach that is to \textit{identify the support of the distribution $L$ to yield an estimate of the true Lipschitz constant of $\phi(\cdot)$.}
	
	\begin{remark}
	Variants of the estimator~\eqref{eq:lipschitz_underestimates} such as $\max_k \ell_{jk}$ have been widely used in the literature to construct algorithms for determining Lipschitz constants, see for example:~\cite{wood1996estimation,strongin1973convergence,calliess2017lipschitz}.\frqed
	\end{remark}
	
	In the literature, common methods of tackling the support estimation problem is by assuming prior knowledge about the exact density of Lipschitz estimates~\cite{calliess2017lipschitz} or using Strongin overestimates of the Lipschitz constant~\cite{strongin1973convergence}. However, we avoid these overestimators because they are provably unreliable, even for globally Lipschitz functions~\cite[Theorem 3.1]{hansen1992using}. Instead, we try to fit the density directly from local estimates and the data in a non-parametric manner using KDE and characteristics of the estimated density. 
	
	\subsection{Plug-in support estimation}
	With a set of $n$ underestimates $\{\ell_{r}\}_{r=1}^n$, we generate an estimate $\hat L_n$ of the true density $L$ using a kernel density estimator
	\begin{equation}\label{eq:kde}
	\hat{L}_n(\ell) = \frac{1}{nh_n}\sum_{r=1}^{n} \mathcal K\left(\frac{\ell - \ell_r}{h_n}\right),
	\end{equation}
	where $\mathcal K : \mathbb R \to \mathbb R$ is a smooth function called the kernel function and $h_n>0$ is the kernel bandwidth. A plug-in estimate of the support $S$ of the true density $L$ is
	\begin{equation}
	\label{eq:estimated_support}
	\hat S_n := \{\ell\in\mathbb R_{\ge 0}: \hat L_n(\ell)\ge\beta_n \},
	\end{equation} 
	where $\beta_n$ is an element of a sequence $\{\beta_n\}$ that converges to zero as $n\to\infty$; this plug-in estimator was proposed in~\cite{cuevas1997plug}.
		
	\subsection{Implementation details}
	Implementing the plug-in estimator involves first constructing a KDE of $L$ with $n$ samples. Then, if one picks $\beta\equiv \beta_n$ small enough, one can easily compute $\hat S$ from~\eqref{eq:estimated_support}. Then \begin{equation}\label{eq:liphat}
	\liphat:= \max(\hat S_n).
	\end{equation} 
	This is a very straightforward operation with the availability of tools like~\texttt{ksdensity} (MATLAB) and the \texttt{KernelDensity} tool in \texttt{scikit-learn} (Python). The pseudocode is detailed herein in Algorithm~\ref{algo:KLL}.
	\begin{algorithm}[!ht]
		\caption{Kernelized Lipschitz Estimation}
		\label{algo:KLL}
		\begin{algorithmic}[1]
			\Require Initial dataset, $\{x_k, \phi(C_q x_k)\}_{k=1}^N$
			\Require Confidence parameter, $0<\beta\ll 1$
			\State $\{q_k, \phi(q_k)\} \leftarrow$ Estimate $C_q$ via ARD
			\For {$k$ in $1,\ldots, N$}
			\For {$j$ in $\{1,\ldots,N\}\setminus k$}
			\State $\ell \leftarrow $ append $\ell_{jk}$ computed by~\eqref{eq:lipschitz_underestimates}
			\EndFor
			\EndFor
			\State $\hat L_n\leftarrow$ KDE with cross-validated $\mathcal K$ and $h$ using $\{\ell_r\}$
			\State $\hat S_n\leftarrow$ compute using~\eqref{eq:estimated_support}
			\State $\liphat\leftarrow \max(\hat S_n)$.
		\end{algorithmic}
	\end{algorithm}

	\begin{remark}
	Note that the true support $S$ is a subset of $\mathbb R_{\ge 0}$. Therefore, when computing the density estimate, this information should be fed into the tool being used. For example, in MATLAB, one has the option $\{$\texttt{'support', 'positive'}$\}$. Essentially, this subroutine transforms the data into the log-scale and estimates the log-density so that upon returning to linear scale, one preserves positivity.
	\end{remark}

	\subsection{Theoretical guarantees}
	We formally describe the density $L$. We consider that the samples $q\in\mathcal Q$ are drawn according to some probability distribution $\mu_0$ with support $\mathbb X$. For any set $S$, suppose that $\mu_0$ can be written as $\mu_0(S) = \int_S \varOmega(q)\, \mathrm{d}\mu(q)$, where $\mu$ is the Lebesgue measure, and $\varOmega$ is continuous and positive on $\mathbb X$. Let $\mu_{X}$ denote the product measure $\mu_0\times\mu_0$ on $\mathbb X\times \mathbb X$. Since $\mu_0$ is absolutely continuous with respect to the Lebesgue measure, $\mu_X$ assigns zero mass on the diagonal $\{(q,q):q\in\mathbb X\}$. The cumulative distribution function for $L$ is then given by 
	\[
	\tilde L(\lambda) = \mu_X \left( \left\{(q_1,q_2): q_1 \neq q_2, \frac{|\phi(q_1) - \phi(q_2)|}{\|q_1-q_2\|} \leq \lambda \right\} \right).
	\]
	Since $\tilde L$ is non-decreasing, $L$ exists almost everywhere by Lebesgue's theorem for differentiability of monotone functions, and $L$'s support is contained within $[0, \lip]$ because of~\eqref{eq:lipschitz_underestimates}.
	
	We investigate the worst-case sample complexity involved in overestimating $\lip$ under the following mild assumption.
	\begin{assumption}\label{asmp:phi_differentiable}
		The nonlinearity $\phi(\cdot)$ is twice continuously differentiable, that is, $\phi(\cdot)\in \mathcal C^2$.\frqed
	\end{assumption}
	
	\begin{lemma}\label{lem:some_q_equals_lip}
	Suppose that Assumptions~\ref{asmp:quality_of_q_samples} and~\ref{asmp:phi_differentiable} hold. Then there exists some $q^\ast\in\mathcal Q$ such that $\|\nabla \phi(q^\ast)\|=\lip$.
	\end{lemma}
	\begin{proof}
	Suppose $\{(q_1^k, q_2^k)\}_{k=1}^\infty$ denotes a sequence of paired samples in $\mathcal Q$ such that $|\phi(q_1^k)-\phi(q_2^k)|/\|q_1^k-q_2^k\|\to \lip$ as $k\to\infty$. Since $\mathcal Q$ is the convex hull of finitely many samples, it is compact, so we can choose a subsequence of $\{q_1^k, q_2^k\}_{k=1}^\infty$ that converges to $(q_1^\infty, q_2^\infty)$ where both limits are in $\mathcal Q$. If $q_1^\infty = q_2^\infty$, then a Taylor expansion estimate implies $\|\nabla \phi(q_1^\infty)\|\ge \lip$. Since $\lip$ is an upper bound of $\|\nabla\phi\|$ at any sample in $\mathcal Q$, $\|\nabla \phi(q_1^\infty)\|= \lip$ and $q^\ast = q_1^\infty$. If $q_1^\infty\neq q_2^\infty$, then the result follows by applying the mean value theorem to $$\varphi(t)=\phi\left(q_1^\infty + t(q_2^\infty - q_1^\infty)\right) - \phi(q_1^\infty)$$ for $t\in [0,1]$, for which $\varphi(0)=0$ and $\varphi(1)=\lip\|q_1^\infty - q_2^\infty\|$. Also, $\mathrm{d}\varphi/\mathrm{d}t = \big(\nabla\phi(q_1^\infty + t(q_2^\infty - q_1^\infty))\big)^\top (q_2^\infty - q_1^\infty)$. Since $\|\nabla \phi\|\le\lip$, this implies $|\mathrm{d}\varphi/\mathrm{d}t|\le \lip\|q_2^\infty - q_1^\infty\|$. Reordering $q_1^\infty$ and $q_2^\infty$ if needed, we have
	\[
	\lip\|q_2^\infty - q_1^\infty\| = \varphi(1) = \int_{0}^1 \big(\mathrm{d}\varphi/\mathrm{d}t\big) \,\mathrm{d}s \le \lip \|q_2^\infty - q_1^\infty\|.
	\]
	Hence, the rightmost inequality must be an equality, which implies that $\textrm{d}\varphi(s)/\textrm{d}t = \lip\|q_1^\infty - q_2^\infty\|$ for all $s$. That is, if the Lipschitz constant is attained with $q_1^\infty\neq q_2^\infty$, then $\phi(\cdot)$ restricted to the segment connecting $q_1^\infty$ and $q_2^\infty$ is linear with slope $\lip$. This concludes the proof.\frQED
	\end{proof}

	Lemma~\ref{lem:some_q_equals_lip} enables the worst-case complexity result described in the following theorem.
	\begin{theorem}\label{thm:worst_case}
	Let 
	$\varphi'(q_1, q_{-1}) = |\phi(q_1)-\phi(q_{-1})|/\|q_1-q_{-1}\|$,
	and suppose that Assumptions~\ref{asmp:quality_of_q_samples} and~\ref{asmp:phi_differentiable} hold. There exists $C_0 > 0$ such that for all sufficiently small $\delta > 0$ and and any set $\{q_j\}_{j=1}^n$ of $n$ uniform random samples in ${\mathbb X}$, the probability that some pair $q_+, q_- \in \{q_j\}$ gives the Lipschitz estimate 
$$ \phi'(q_+, q_-) \geq (1-\delta) {\mathcal L_\phi^*} - C_0 \delta$$
is at least $1 - \epsilon(n, \delta)$.  Here $\epsilon(n, \delta) \leq 3\exp(-n \kappa \delta^{2 n_q-1})$, where $\kappa$ is a constant depending on $n_q$.
	\end{theorem}
	\begin{proof}
		By Lemma~\ref{lem:some_q_equals_lip}, there exists at least one $q^\ast$ such that $\|\nabla \phi(q^\star)\|=\lip$. For the worst-case analysis, suppose this occurs only at a single sample, $q^\star$. A Taylor expansion at $q^\star$ yields
		\begin{equation} \label{eqn:phi}
		\phi(q^\star + q) = \phi(q^\star) + \nabla \phi(q^\star)^\top q + \mathfrak R(q),
		\end{equation}
		where $\mathfrak R$ is a remainder term with $|\mathfrak R(q)| \leq C_{\mathfrak R} \|q\|^2$ when $\|q\| \leq \eta$, for some $C_{\mathfrak R}>0$ and $\eta > 0$.  Note  that 
		$$  \nabla \phi(q^\star)^\top q = \| \nabla \phi(q^\star)\| \|q\| \cos \theta, $$
		where $\theta$ is the angle between  $\nabla \phi(q^\star)$ and $q$.  To obtain a good estimate of  $\| \nabla \phi(q^\star)\|$, one needs to sample two points in this neighborhood, one point $q_+$ with $\cos \theta\approx 1$ and a second point $q_{-}$ with $\cos \theta \approx -1$.  Each of these conditions defines a cone.    Regarding one of these cones in cylindrical coordinates $0 \leq \ell \leq \eta$ and $\|y\| \leq \chi \ell$ for $\chi = \tan \theta$, we can integrate the $n_q-1$ dimensional volume to get the volume of this cone as $C_0 \chi^{n_q-1} \eta^{n_q}$ for some dimension-dependent constant $C_0$.   A calculation shows that $\cos \theta \geq 1 - \chi^2/2$ for small $\theta$.  With $q_+$ and $q_{-}$ as one sample from each cone, we have that $q_+ - q_{-}$ is contained in the cone $\|y\| \leq \chi \ell$, and so we can use \eqref{eqn:phi} to approximate
		\begin{align*}
		&| \phi(q^\star+q_+) - \phi(q^\star+q_{-}) | \\
		&= | \nabla \phi(q^\star)^T (q_+-q_{-}) +\mathfrak R(q_+) - \mathfrak R(q_{-}) |\\
		&\geq \|\nabla \phi(q^\star)\| \|q_+ - q_{-}\| \left(1 - \frac{\chi^2}{2}\right) - | \mathfrak R(q_+) - \mathfrak R(q_{-}) |.
		\end{align*}
		Dividing by $ \|q_+ - q_{-}\|$ and using the defining property of $q^\star$ gives 
		\begin{equation} \label{eqn:phi2}
		\phi'(q_+, q_{-}) \geq  \left(1 - \frac{\chi^2}{2}\right) {\mathcal L}^*_{\phi}- \frac{ | \mathfrak R(q_+) - \mathfrak R(q_{-}) |}{ \|q_+ - q_{-}\|}.
		\end{equation}
		
		Subsequently, one can decompose $q_+ = q_+^\parallel + q_+^\perp$, with $q_+^\parallel$ parallel to $\nabla \phi(q^\star)$ and $q_+^\perp$ perpendicular and satisfying $\|q_+^\perp\| \leq \chi |q_+^\parallel|$. Identical arguments can be used to infer  $\|q_{-}^\perp\| \leq \chi |q_{-}^\parallel |$, and hence,
		$ \| q_+ - q_{-}\| \geq \|q_+^\| - q_{-}^\|\| - \|q_+^\perp - q_{-}^\perp\|. $
		Since $q_+$ and $q_{-}$ are chosen from opposite cones, we have $\|q_+^\| - q_{-}^\|\| = \|q_+^\| \| +  \|q_{-}^\|\|$.  Using $\|q^\perp\| \leq \chi \|q^\|\|$ and $\cos \theta \geq 1-\chi^2/2$, we have $\|q_+ - q_{-}\| \geq (\|q_+\| +  \|q_{-}\|)(1-\chi)\left(1 - \chi^2/2\right)$.
		Hence,
		\begin{align*}
		\frac{ |\mathfrak R(q_+) - \mathfrak R(q_{-}) |}{ \|q_+ - q_{-}\|} &\leq \frac{C_{\mathfrak R} (\|q_+\|^2 + \|q_{-}\|^2)}{ (\|q_+\| +  \|q_{-}\|)(1-\chi)\left(1 - \chi^2/2\right)}\\
		& \leq \frac{2 C_{\mathfrak R} \max\{\|q_+\|, \|q_{-}\|\}^2}{\max\{\|q_+\|, \|q_{-}\|\} (1-\chi)\left(1 - \chi^2/2\right)}\\
		&\leq \frac{2 C_{\mathfrak R} \eta (1 + \chi^2)^{1/2}}{ (1-\chi)\left(1 - \chi^2/2\right)}.
		\end{align*}
		By combining the aforementioned inequality with ~\eqref{eqn:phi2}, one obtains
		\[
		\varphi'(q_+,q_-)\ge \left(1-\frac{\chi}{2}\right)\lip - \frac{2 C_{\mathfrak R} \eta (1 + \chi^2)^{1/2}}{ (1-\chi)\left(1 - \chi^2/2\right)}.
		\]
		Set $\delta = \chi/2$ and take $\eta = \delta$.  Then there exists $C_0 > 0$ such that for all sufficiently small $\delta$, 
		\begin{equation} \label{eqn:worst-case}
		\varphi'(q_+,q_-) \ge \left(1-\delta\right)\lip - C_0\delta,
		\end{equation}
		which implies that $\varphi'\to \lip$ as $\delta\to 0$.
		
		From the assumption on uniformly drawn samples, the probability of sampling in one of the $(\chi, \eta)$ cones is,
		$$ \int_0^\eta \kappa (\chi r)^{n_q-1} \,\mathrm{d}r = \frac{\kappa \chi^{n_q-1} \eta^{n_q}}{n_q} $$ 
		for some $\kappa>0$ that depends on $n_q$.  Using $\delta = \eta = \chi/2$ and absorbing the factors of 2 and $1/n_q$ into $\kappa$ yields $\kappa \delta^{2n_q - 1}$.  
		
		Let $\mathfrak X_{1 \pm}$ be the event of sampling at least one point in the $(\chi,\eta)$ cone as above, and let $\mathfrak X_{0 \pm}$ be the event of sampling nothing in the $(\chi,\eta)$  cone.  The probability of sampling at least one of each of the points $q_+$ and $q_{-}$ just described is,
		\begin{align*}
		1 &- P(\mathfrak X_{0+} \cap \mathfrak X_{0-}) -   P(\mathfrak X_{0+} \cap \mathfrak X_{1-}) - P(\mathfrak X_{1+} \cap \mathfrak X_{0-})\\
		& \geq 1 - (1 - 2 \kappa \delta^{2 n_q-1})^n - 2(1 - \kappa \delta^{2 n_q-1})^n,
		\end{align*}
		where the factor 2 before $\kappa$ in the second term comes from the fact that both cones are excluded and they are disjoint, and the 2 before the third term comes by combining the final two terms in the first expression.    
		
		By using the fact that for any $\epsilon' \in (0,1)$ and $n > 0$ the inequality $(1-\epsilon')^n \leq \exp(-n \epsilon')$ holds, we can conclude that,
		\begin{align*}
		1 &- P(\mathfrak X_{0+} \cap \mathfrak X_{0-}) -   P(\mathfrak X_{0+} \cap \mathfrak X_{1-}) - P(\mathfrak X_{1+} \cap \mathfrak X_{0-})\\
		&\geq 1 - \exp(-2 n \kappa \delta^{n_q-1}) - 2\exp(-n \kappa \delta^{n_q-1}) \geq 1 - \epsilon,
		\end{align*}
		for any given $\epsilon>0$. The latter can be ensured by choosing $n$ large enough.
	    This gives a lower bound on the probability of obtaining~\eqref{eqn:worst-case} and hence, the desired result.  \frQED
	\end{proof}

	\subsection{Benchmarking the Lipschitz estimator}
	Our Lipschitz estimator is tested on well-studied benchmark examples studied previously in~\cite{wood1996estimation,calliess2014conservative}:
	the benchmark functions are described in Table~\ref{tab:benchmark_functions} along with their domains and true local Lipschitz constants. Note that all the functions are not globally Lipschitz (e.g. $\phi_2$), not differentiable everywhere (e.g. $\phi_1$, $\phi_4$), and, in the special case of $\phi_4$, specifically constructed to ensure that naive overestimation of $\lip$ using Strongin methods provably fails~\cite{hansen1992using}. To evaluate the proposed Lipschitz estimator, we vary the number of data points $N$ and the confidence parameter $\beta$. Over 100 runs, we report mean $\pm$ one standard deviation of the following quantities: the time required by our learning algorithm, the estimated Lipschitz constant $\liphat$, and the error $\liphat-\lip$ (which should be positive when we overestimate $\lip$).
	\begin{table}[!ht]
		\footnotesize
		\def\arraystretch{1.15}
		\centering
		\caption{Kernelized Lipschitz Learning of benchmark functions}
		\begin{threeparttable}			
			\label{tab:benchmark_functions}
			\begin{tabular}{c|c|c|c|c}
				\hline
				$n$ & $\log_{10}\beta$	& Time [s] & $\liphat$ (mean $\pm$ stdev)& OE$^\ddag$? \\
				\hline
				\rowcolor{RoyalBlue}
				\multicolumn{5}{c}{$\phi_1 = |\cos(\pi x)|$, $\lip = 3.141$ on $[-\pi,\pi]$} \\
				\hline
				$100$  & $-2$ & 0.596 $\pm$ 0.127 & 3.361 $\pm$ 0.093 &  $\checkmark$\\
				$100$  & $-4$ &	0.610 $\pm$ 0.122 &	3.528 $\pm$ 0.177 &	 $\checkmark$\\
				$500$  & $-2$ &	2.463 $\pm$ 0.432 &	3.252 $\pm$ 0.083 & $\checkmark$\\
				$500$  & $-4$ &	2.438 $\pm$ 0.427 &	3.369 $\pm$ 0.158 &	 $\checkmark$\\
				\hline
				\rowcolor{RoyalBlue}
				\multicolumn{5}{c}{$\phi_2 = x-x^3/3$, $\lip = 1.000$ on $[-1,1]$} \\
				\hline
				$100$	& $-2$ & 0.455 $\pm$ 0.123 & 1.018 $\pm$ 0.011 & $\checkmark$\\
				$100$	& $-4$ & 0.459 $\pm$ 0.114 & 1.030 $\pm$ 0.020 &  $\checkmark$\\
				$500$   & $-2$ & 1.656 $\pm$ 0.218 & 1.005 $\pm$ 0.004 &  $\checkmark$\\
				$500$   & $-4$ & 1.585 $\pm$ 0.208 & 1.010 $\pm$ 0.009 &  $\checkmark$\\
				\hline
				\rowcolor{RoyalBlue}
				\multicolumn{5}{c}{$\phi_3 = \sin(x)+\sin(2x/3)$, $\lip = 1.667$ on $[3.1,20.4]$} \\
				\hline
				$100$ &	$-2$ & 0.556 $\pm$ 0.121 & 1.780 $\pm$ 0.074 & $\checkmark$\\
				$100$ & $-4$ & 0.547 $\pm$ 0.124 & 1.923 $\pm$ 0.166 &  $\checkmark$\\
				$500$ & $-2$ & 1.826 $\pm$ 0.224 & 1.684 $\pm$ 0.010 & $\checkmark$\\ 
				$500$ & $-4$ & 1.821 $\pm$ 0.221 & 1.720 $\pm$ 0.002 &  $\checkmark$\\
				\hline
				\rowcolor{RoyalBlue}
				\multicolumn{5}{c}{$\phi_4 =$ Hansen test function from~\cite{wood1996estimation}, $\lip = 8.378$ on $[0,1]$}\\
				\hline
				$100$ &	$-2$ & 0.450 $\pm$ 0.117 & 8.969 $\pm$ 0.262 &  $\checkmark$\\
				$100$ &	$-4$ & 0.488 $\pm$ 0.123 & 9.401 $\pm$ 0.476 &  $\checkmark$\\
				$500$ &	$-2$ & 1.921 $\pm$ 0.138 & 8.507 $\pm$ 0.046 &  $\checkmark$\\
				$500$ &	$-4$ & 1.923 $\pm$ 0.130 & 8.707 $\pm$ 0.103 &  $\checkmark$\\
				\hline 
				\rowcolor{RoyalBlue}
				\multicolumn{5}{c}{$\phi_5 = \max\{1-3\sin(x), \exp(-\sin(x))\}$, $\lip = 3.0$ on $[-10,10]$} \\
				\hline
				$100$ &	$-2$ & 0.546 $\pm$ 0.061 & 3.139 $\pm$ 0.051 &  $\checkmark$\\
				$100$ &	$-4$ & 0.612 $\pm$ 0.066 & 3.200 $\pm$ 0.087 &  $\checkmark$\\
				$500$ &	$-2$ & 1.893 $\pm$ 0.245 & 3.043 $\pm$ 0.014 &  $\checkmark$\\
				$500$ &	$-4$ & 1.989 $\pm$ 0.230 & 3.104 $\pm$ 0.024 &  $\checkmark$\\
				\hline
			\end{tabular}
			\begin{tablenotes}
				\item $\ddag$ OE indicates an overestimate of the true Lipschitz constant, that is, $\min\{\liphat\} > \lip$.
			\end{tablenotes}
		\end{threeparttable}
	\end{table}

	The final column of Table~\ref{tab:benchmark_functions} reveals an important empirical detail: all our estimates of $\lip$ are overestimates for $\beta\le 0.01$ and $n\ge 100$. This is a critical advantage of our proposed approach, because an overestimate will enable us to provide stability and constraint satisfaction guarantees about the data-driven controller, as we will discuss in subsequent sections. 
	Furthermore, the estimation error is small for every test run, and as expected, the error increases as $\beta$ decreases, because a smaller value of $\beta$ indicates the need for greater confidence, which results in more conservative estimates. 
	
	\section{Safe Initialization in ADP}\label{sec:control}
	In this section, we begin by reviewing a general ADP procedure, and then explain how to safely initialize unconstrained, as well as input-constrained ADP.
	\subsection{Unconstrained ADP}
	Recall the optimal value function given by~\eqref{eq:J_opt} and the optimal control policy~\eqref{eq:u_opt}. From the Bellman optimality principle, we know that the discrete-time Hamilton-Jacobi-Bellman equations are given by
	\begin{align}
	\label{eq:bellman_cost}
	J_\infty(x_t) &= \inf_{u\in\mathfrak U} \left(\mathcal U(x_t, u(x_t)) + \gamma J_\infty(x_{t+1})\right),\\
	\label{eq:bellman_control_policy}
	u_\infty(x_t) &= \arginf_{u\in\mathfrak U} \left(\mathcal U(x_t, u(x_t)) + \gamma J_\infty(x_{t+1})\right),
	\end{align}
	where $J_\infty(x_t)$ is the optimal value function and $u_\infty(x_t)$ is the optimal control policy.
	The key operations in ADP methods~\cite{lewis2012reinforcement} involve setting an admissible control policy $u_0(x)$ and then iterating the policy evaluation step
	\begin{subequations}
	\label{eq:policy_iteration}
	\begin{equation}
	\label{eq:policy_evaluation_ADP}
	\mathcal J_{k+1}(x_t) = \mathcal U\big(x_t, u_k(x_t)\big) + \mathcal \gamma\mathcal J_{k+1}(x_{t+1})
	\end{equation}
	and the policy improvement step
		\begin{equation}
		\label{eq:policy_improvement_ADP}
		u_{k+1}(x_t) = \argmin_{u(\cdot)} \left(\mathcal U\big(x_t, u(x_t)\big)+\gamma\mathcal J_{k+1}(x_{t+1})\right)
		\end{equation}
		\end{subequations}
	until convergence.
	
	\subsubsection{Semidefinite programming for safe initial control policy}
	Recall the following definition.
	\begin{definition}
		The equilibrium point $x=0$ of the closed-loop system~\eqref{eq:cl_sys} is globally exponentially stable with a decay rate $\alpha$ if there exist scalars $C_0>0$ and $\alpha\in (0,1)$ such that
		$\|x_t\| \le C_0\alpha^{(t-t_0)}\|x_0\|$
		for any $x_0\in\mathbb R^{n_x}$.\frqed
	\end{definition}

	Conditions for global exponential stability (GES) of the equilibrium state, adopted from \cite{khalil2015nonlinear}, is provided next.
	\begin{lemma}\label{lem:CtrlLem_2}			
	Let $V(\cdot,\cdot):[0,\infty)\times \mathbb R^{n_x}\to \mathbb R$ be a continuously differentiable function such that
	\begin{subequations}
		\label{eq:CtrlLem2}
		\begin{align}
		\label{eq:CtrlLem2_a}
		\gamma_1 \|x\|^2 \le V(t, x_t) &\le \gamma_2 \|x\|^2\\    				\label{eq:CtrlLem2_b}
		V(t+1, x_{t+1}) - V(t, x_t) &\le -(1-\alpha^2) V(t, x_t),
		\end{align}
	\end{subequations}
	for any $t\ge t_0$ and $x\in  \mathbb R^{n_x}$ along the trajectories of the system
	\begin{equation}
	\label{eq:sys}
	x^+ = \varphi(x),
	\end{equation} 
	where $\gamma_1$, $\gamma_2$, and $\alpha$ are positive scalars, and $\varphi(\cdot)$ is a nonlinear function. Then the equilibrium state $x=0$ for the system~\eqref{eq:sys} is GES with decay rate $\alpha$.\frqed
	\end{lemma}
	
	The following design theorem provides a method to construct an initial linear stabilizing policy $u_0(x) = K_0x$ such that the origin is a GES equilibrium state of the closed-loop system~\eqref{eq:cl_sys}.
	\begin{theorem}\label{thm:control_1}
	Suppose that Assumptions 1--2 hold, and that there exist matrices $P=P^\top\succ 0\in\mathbb R^{n_x\times n_x}$, $K_0\in\mathbb{R}^{n_u\times n_x}$, and scalars $\alpha\in (0,1)$, $\nu>0$ such that 
			\begin{align}
			\label{eq:thm_control_1}
				\Psi + \Gamma^\top \mathcal M \Gamma &\preceq 0,
			\end{align}
			where
			\begin{align*}
				\Psi &= \begin{bmatrix}
					(A+BK_0)^\top P (A+BK_0) - \alpha^2 P & \star \\ G^\top P (A+BK_0) & G^\top P G
				\end{bmatrix},\\
				\Gamma &= \begin{bmatrix}C_q & 0 \\ 0 & I\end{bmatrix},\; \text{and} \;
				\mathcal M = \begin{bmatrix}\nu^{-1} (\lip)^2 I & 0 \\ 0 & -\nu^{-1} I\end{bmatrix}.
			\end{align*}
	Then the equilibrium $x=0$ of the closed-loop system~\eqref{eq:cl_sys} is GES with decay rate $\alpha$.
	\end{theorem}
	
	\begin{proof}
		Let $V = x^\top P x$. Then~\eqref{eq:CtrlLem2_a} in Lemma~\ref{lem:CtrlLem_2} is satisfied with $\gamma_1 = \lambda_{\min}(P)$ and $\gamma_2 = \lambda_{\max}(P)$.
		Let $\Delta V = V^+ - V$. Note that
		\begin{align*}
		V^+ &= (x^+)^\top P x^+\\
		&= \left((A+BK_0) x + G \phi\right)^\top P \left((A+BK_0) x + G \phi\right)\\
		&= x^\top(A+BK_0)^\top P (A+BK_0)x  \\ &\qquad + 2x^\top (A+BK_0)^\top P G \phi + \phi^\top G^\top P G \phi.
		\end{align*}
		Therefore,
		\begin{align*}
		\begin{bmatrix} x \\ \phi\end{bmatrix}^\top \Psi \begin{bmatrix} x \\ \phi\end{bmatrix} &= x^\top (A+BK_0)^\top P (A+BK_0) x - \alpha^2 x^\top P x \\&\qquad+ 2x^\top (A+BK_0)^\top P G \phi + \phi^\top G^\top P G \phi\\
		&=V^+ - \alpha^2 V = \Delta V + (1-\alpha^2)V,
		\end{align*}
		and
		\begin{align*}
		\begin{bmatrix} x \\ \phi\end{bmatrix}^\top \Gamma^\top \mathcal M \Gamma \begin{bmatrix} x \\ \phi\end{bmatrix} &= \begin{bmatrix} q \\ \phi\end{bmatrix}^\top \mathcal M \begin{bmatrix} q \\ \phi\end{bmatrix}= \nu\left((\lip)^2 q^\top q - \phi^\top \phi\right).
		\end{align*}
		Thus, pre- and post-multiplying~\eqref{eq:thm_control_1} with $\begin{bmatrix} x & \phi\end{bmatrix}^\top$ and its transpose, respectively, we get
		\[
		\Delta V + (1-\alpha^2) V + \nu\left((\lip)^2 q^\top  q - \phi^\top \phi\right) \le 0.
		\]
		By inequality~\eqref{eq:PhiLipschitz} in Assumption~\ref{asmp:nonlin_knowledge} and recalling that $\phi(0) = 0$, we get
		$
		(\lip)^2 q^\top  q - \phi^\top \phi \ge 0.
		$
		Since $\nu>0$, this implies
		$
		\Delta V + (1-\alpha^2) V \le 0,
		$
		which is identical to~\eqref{eq:CtrlLem2_b}.\frQED
	\end{proof}
	Note that we do not need to know $\phi(\cdot)$ to satisfy conditions~\eqref{eq:thm_control_1}. Instead, Theorem~\ref{thm:control_1} provides conditions that leverage matrix multipliers similar to those described in~\cite{chakrabarty2017state}. 

	We shall now provide LMI-based conditions for computing the initial control policy $K_0$, the initial domain of attraction $P$ and $\nu$ via convex programming.
	\begin{figure*}[!t]
	\small
	\rule[1ex]{\linewidth}{0.25pt}
	\begin{align}\label{longeq:1}
	\begin{bmatrix}
	-\alpha^2 P & 0 \\ 0 & -\nu^{-1} I
	\end{bmatrix}  + \left[\begin{array}{cc}
	(A+BK_0)^\top & C_q^\top \\ G^\top & 0
	\end{array}\right]\left[\begin{array}{cc}
	P & 0 \\ 0 & \liphat^2 \nu^{-1} I
	\end{array}\right]\left[\begin{array}{cc}
	(A+BK_0)^\top & C_q^\top \\ G^\top  & 0
	\end{array}\right]^\top &\preceq 0\\
	\label{longeq:2}
	\begin{bmatrix}
	(A+BK_0)^\top P (A+BK_0) - \alpha^2 P & (A+BK_0)^\top PG \\ G^\top P(A+BK_0) & -\nu^{-1}I + G^\top P G
	\end{bmatrix} +
	\begin{bmatrix} C_q \\ 0 \end{bmatrix} (\liphat)^2\nu^{-1} I \begin{bmatrix} C_q \\ 0 \end{bmatrix}^\top &\preceq 0
	\end{align}
	\rule[1ex]{\linewidth}{0.25pt}
	\end{figure*}
	
	\begin{theorem}
		\label{thm:control_2}		
		Fix $\alpha\in (0,1)$ and $\liphat$ obtained via~\eqref{eq:liphat}. If there exist matrices $S=S^\top\succ 0$, $Y$, and a scalar $\nu>0$ such that the LMI conditions
			\begin{align}
				\label{eq:mainLMI_d}
				\begin{bmatrix}
					-\alpha^2 S & \star & \star & \star \\
					0 & - \nu I & \star & \star\\
					AS + BY & \nu GS & -S & \star\\
					\liphat C_q S & 0 & 0 & -\nu I
				\end{bmatrix} &\preceq 0
			\end{align}
		are satisfied, then the matrices $K_0=YS^{-1}$, $P=S^{-1}$ and scalar $\nu$ satisfy the conditions~\eqref{eq:thm_control_1} with the same $\alpha$ and $\liphat$.
	\end{theorem}
	
	\begin{proof}
		
		A congruence transformation of~\eqref{eq:mainLMI_d} with the matrix $\mathrm{blkdiag}\left(\begin{bmatrix}
		P & \nu^{-1}\,I & P & I 
		\end{bmatrix}\right)$ and substituting $S$ with $P^{-1}$ and $Y$ with $K_0P^{-1}$ yields
		\[
		\left[\begin{array}{cc|cc}
		-\alpha^2 P & \star & \star & \star \\
		0 & - \nu^{-1} I & \star & \star\\\hline
		A+BK_0 & G & - P & \star\\
		\liphat C_q & 0 & 0 & -\nu I
		\end{array}\right] \preceq 0.
		\]
		Taking the Schur complement with the submatrices shown by the guidelines in the above inequality, we get~\eqref{longeq:1}. Since $\nu >0$, taking the Schur complement again yields~\eqref{longeq:2}
		which can be rewritten as
		\begin{align*}
		\Psi -  \begin{bmatrix} 0 \\ I \end{bmatrix} \nu^{-1} I \begin{bmatrix} 0 \\ I \end{bmatrix}^\top + \begin{bmatrix} C_q\\ 0 \end{bmatrix} (\liphat)^2\nu^{-1} I \begin{bmatrix} C_q \\ 0 \end{bmatrix}^\top &\preceq 0
		\end{align*}
		which is exactly~\eqref{eq:thm_control_1}. Thus, the conditions~\eqref{eq:thm_control_1} and~\eqref{eq:mainLMI_d} are equivalent.\frQED
	\end{proof}
	
	Empirically, we observe that our proposed kernelized Lipschitz learner  typically provides overestimates of $\lip$ (see Appendix). 
	A benefit of overestimating $\lip$ is that admissibility of the control policy is ensured. This is demonstrated by the following result.
	\begin{theorem}\label{prop:1}
		Let $(P, K_0, \nu, \alpha)$ be a feasible solution to the conditions~\eqref{eq:thm_control_1} with an overestimate of the Lipschitz constant $\liphat>\lip$. Then $(P, K_0, \nu, \alpha)$ is also a feasible solution to the conditions~\eqref{eq:thm_control_1}.
	\end{theorem}
	
	\begin{proof}
		Let $\delta L = \liphat - \lip$. Since $\liphat$ is an overestimator of $\lip$, $\delta L>0$. Since $(P, K, \nu,\alpha)$ is a feasible solution to~\eqref{eq:thm_control_1} with $\liphat$, it satisfies
		\[
		\Psi + \Gamma^\top  \begin{bmatrix} -\nu^{-1}(\lip+\delta L)^2 I & 0\\ 0 & I\end{bmatrix} \Gamma \preceq 0,
		\]
		which can be written as
		\[
		\Psi + \Gamma^\top \mathcal M \Gamma + \Gamma^\top  \underbrace{\begin{bmatrix} -\nu^{-1}(2\lip\delta L + \delta L^2) I & 0\\ 0 & 0\end{bmatrix}}_{:=\delta {\mathcal M}} \Gamma \preceq 0.
		\]
		As $\nu>0$, we infer that $\delta\mathcal M\preceq 0$, hence $\Gamma^\top \delta\mathcal M \Gamma \preceq 0$.
		Therefore, $\Psi + \Gamma^\top\mathcal M\Gamma \preceq 0$. Since the other conditions in~\eqref{eq:thm_control_1} are independent of $\lip$, the other conditions are automatically satisfied. This concludes the proof.\frQED
	\end{proof}
	
	Theorem~\ref{prop:1} indicates that if our learned $\liphat$ is an overestimate of $\lip$, and we use $\liphat$ to obtain a safe stabilizing control policy, then this is also a safe stabilizing control policy for the true system~\eqref{eq:true_sys}.
	Having a feasible solution to~\eqref{eq:thm_control_1} with an underestimator of $\lip$ is not sufficient to guarantee a feasible solution for the true Lipschitz constant, because $\delta\mathcal M$ may not be negative semi-definite in that case. Of course, extremely conservative overestimates of $\liphat$ will result in conservative control policies or result in infeasibility. In our proposed approach, we have observed that the confidence parameter $\beta$ dictates the conservativeness of the overestimate; that is $\beta\to 1$ makes the estimate $\liphat$ more conservative.
	
	\subsubsection{Safely initialized PI} 
	
	We begin by proving the following critical result.
	\begin{theorem}\label{thm:admissible_1}
	Let $\mathcal U(x,u)$ be defined as in~\eqref{eq:cost}. If $K_0$ is obtained by solving~\eqref{eq:mainLMI_d} for $\liphat\ge \lip$, then the initial control policy $u_0=K_0x$ is an admissible control policy on $\mathbb R^{n_x}$.
	\end{theorem}
	\begin{proof}
	Clearly, $u_0$ is continuous, and (by Theorem~\ref{thm:control_1} and~\ref{thm:control_2}) is a stabilizing control policy for~\eqref{eq:true_sys}. It remains to show that the cost induced by $u_0$ is finite. Since $u_0$ is stabilizing and $\liphat\ge \lip$, we know that $\|x_t\|\to 0$ as $t\to\infty$, which implies $u_0\to 0$ and, by therefore, $\mathcal U(x_t,u_t)\to 0$ as $t\to\infty$. Since $\mathcal U(x_t,u_t)$ converges to a finite limit, $\mathcal U(x_t,u_t)$ is bounded for all $t\ge 0$. Therefore, any partial sum $\sum_{t=0}^{t'} \mathcal U(x_t,u_t)$ is bounded and monotonic; that is, $\mathcal J$ converges to a finite limit.\frQED
	\end{proof}
	Admissibility of $u_0$ for the specific linear quadratic regulator (LQR) cost function follows directly from Theorem~\ref{thm:admissible_1}.
	\begin{corollary}
	Let 
	\begin{equation}\label{eq:LQR_cost}
	\mathcal U(x_t,u_t)=x_t^\top Q x_t + u_t^\top R u_t 
	\end{equation}
	for some matrices $Q=Q^\top\succeq 0$ and $R=R^\top\succ 0$. Then the initial control policy $u_0=K_0x$ obtained by solving~\eqref{eq:mainLMI_d} is an admissible control policy on $\mathbb R^{n_x}$.
	\end{corollary}

	Now that we know $u_0=K_0x$ is an admissible control policy, we are ready to proceed with the policy iteration steps~\eqref{eq:policy_iteration}. Typically, an analytical form of $\mathcal J_k$ is not known \textit{a priori}, so we resort to a shallow neural approximator/truncated basis expansion for fitting this function, assuming $\mathcal J_k$ is smooth for every $k\in\mathbb N\cup\{\infty\}$. Concretely, we represent the value function and cost functions as:
	\begin{equation}
	\label{eq:Jk}
	\mathcal J_k(x) :=\omega_k^\top \psi(x)
	\end{equation}
	where $\psi_0(\cdot):\mathbb {R}^{n_x}\to \mathbb R^{n_0}$ denotes the set of differentiable basis functions (equivalently, hidden layer neuron activations) and $\omega:\mathbb{R}^{n_0}$ is the corresponding column vector of basis coefficients (equivalently, neural weights). 
	
	It is not always clear how to initialize the weights of the neural approximators~\eqref{eq:Jk}. Commonly, small random numbers drawn from a uniform distribution are used~\cite{liu2014policy}, but there is no safety guarantee associated with random initialization. We propose initializing the weights as follows. Since our initial Lyapunov function is quadratic, we include the quadratic terms of the components of $x$ to be in the basis $\psi(x)$. Then we can express the initial Lyapunov function $x^\top P x$ obtained by solving~\eqref{eq:mainLMI_d} with appropriate weights in the $\psi(x)$, respectively, setting all other weights to be zero.
	With the approximator initialized as above, the policy evaluation step~\eqref{eq:policy_evaluation_ADP} is replaced by
	\begin{subequations}
	\label{eq:policy_iteration_nn}
	\begin{equation}
	\label{eq:policy_evaluation_nn}
	\omega_{k+1}^\top \big(\psi(x_t)- \gamma\psi(x_{t+1})\big) = \mathcal U\left(x_t, u_k(x_t)\right),
	\end{equation}
	from which one can solve for $\omega_{k+1}$ recursively via
	\[
	\omega_{k+1} = \omega_{k} - \eta_k\varphi_k\left(\omega_{k}^\top \varphi_k - \mathcal U\left(x_t, u_k(x_t)\right)\right),
	\] 
	where $\eta_k$ is a learning rate parameter that is usually selected to be an element from the sequence $\{\eta_k\}\to 0$ as $k\to\infty$, and $	
	\varphi_k = \psi(x_t) - \gamma \psi(x_{t+1})$.
	Subsequently, the policy improvement step~\eqref{eq:policy_improvement_ADP} is replaced by
	\begin{equation*}
	u_{k+1} = \argmin_{u(\cdot)} \left(\mathcal U\left(x_t, u(x_t)\right) + \gamma \omega_{k+1}^\top \psi(x_{t+1})\right).
	\end{equation*}
	This minimization problem is typically non-convex and therefore, challenging to solve to optimality. In some specific cases, one of which is that the cost function is quadratic as described in~\eqref{eq:LQR_cost}, the policy improvement step becomes considerably simpler to execute, namely 
	\begin{equation}
	\label{eq:policy_improvement_nn}
	u_{k+1}(x) = -\frac{\gamma}{2} R^{-1} B^\top \nabla\psi(x)^\top \omega_{k+1}.
	\end{equation}
	\end{subequations}
	This can be evaluated as $R$ and $B$ are known, and $\psi$ is differentiable and chosen by the user, so $\nabla \psi$ is computable. 
	
	Since we prove that $u_0$ is an admissible control policy, we can use arguments identical to~\cite[Theorem 3.2 and Theorem 4.1]{liu2014policy} to claim that if the optimal value function and the optimal control policy are dense in the space of functions induced by the basis function expansions~\eqref{eq:Jk}, then the weights of the neural approximator employed in the PI steps~\eqref{eq:policy_iteration_nn} converges to the optimal weights; that is, the optimal value function $\mathcal J_\infty$ and the optimal control policy $u_\infty$ are achieved asymptotically. A pseudocode for implementation is provided next.
	
		\begin{algorithm}[!ht]
		\small
		\caption{Safely Initialized PI for discrete-time systems}
		\label{algo:PI}
		\begin{algorithmic}[1]
			\Require Termination condition constant $\epsilon_\mathrm{ac}$
			\Require Historical data $\mathcal D$
			\State Estimate Lipschitz constant $\liphat$ using Algorithm~\ref{algo:KLL}
			\Require Compute stabilizing control gain $K_0$ via SDP~\eqref{eq:mainLMI_d}
			\State Fix admissible control policy $u_0(x)=K_0x$
			\While {$\left\|\mathcal J_k - \mathcal J_{k-1}\right\|\geq \epsilon_\mathrm{ac}$}
			\State Solve for the value $\mathcal J_{k}(x)$ using
			\[
			\mathcal J_{k+1}(x_t)=\mathcal U(x_t,u_k(x_t))+\gamma \mathcal J_{k+1}(x_{t+1}).
			\]
			\State Update the control policy $u_{(k+1)}(x)$ using
			\[
			u_{k+1}(x_t)=\argmin_{u(\cdot)}\big(\mathcal U(x_t,u_k(x_t))+\gamma \mathcal J_{k+1}(x_{t+1})\big).
			\]
			\State $k:=k+1$
			\EndWhile
		\end{algorithmic}
	\end{algorithm}

	\subsection{Input-constrained ADP with safety}
	Herein, we tackle the case when the control input is to be constrained, which is very common in practical applications. We make the following assumption on the constraints.
	\begin{assumption}\label{asmp:polytopic_constraints}
	The control input $u\in\mathbb U$, where
	\begin{equation}\label{eq:constraints}
	\mathbb U = \left\{u\in\mathbb{R}^{n_u}: \xi_i^\top u\le 1\right\},
	\end{equation}
	for $i=1,\ldots, n_{c}$, where $n_c$ is the number of input constraints, and $\xi_i\in\mathbb R^{n_u}$ for every $i$.\frqed
	\end{assumption}
	\begin{remark}
	The matrix inequality~\eqref{eq:constraints} defines a polytopic input constraint set. Clearly, constraints of the form $|u|\le \bar u$ can be written as
	\[
	\begin{bmatrix} \xi_i \\ \xi_{i+1} \end{bmatrix}u = \begin{bmatrix}
	0 & \cdots & 1/\bar u & \cdots & 0\\
	0 & \cdots & -1/\bar u & \cdots & 0
	\end{bmatrix}u\le \begin{bmatrix}
	1\\ 1
	\end{bmatrix},
	\] 
	which is of the form~\eqref{eq:constraints}.\frqed
	\end{remark}
	
	Note that with any control policy $u_0=K_0 x$, the constraint set described in~\eqref{eq:constraints} is equivalent to the set
	\begin{equation}\label{eq:constraints_x}
		\mathcal X = \left\{x\in\mathbb{R}^{n_x}: \xi_i^\top K_0 x \le 1\right\},
	\end{equation}
	for $i=1,\ldots, n_c$.
	Before we state the main design theorem, we require the following result from~\cite[pp. 69]{boyd1994linear}.
	\begin{lemma}\label{lem:ellipsoid}
		The ellipsoid 
		\begin{subequations}
		\begin{equation}\label{eq:ellipsoid_P}
			\mathcal E_P=\{x\in\mathbb R^{n_x}: x^\top P x \le 1\}
		\end{equation} is a subset of $\mathcal X$ if and only if
		\begin{equation}\label{eq:ellipsoid_containment}
			\xi_i K_0^\top P^{-1} K_0\, \xi_i^\top \le 1	
		\end{equation}
		\end{subequations}
		for $i=1,\ldots, n_c$.\frqed
	\end{lemma}

	\subsubsection{Constrained admissible initial control policy and invariant set estimation}
	Since the control input is constrained, we need to characterize an invariant set of the form $\mathcal E_P$ within which all control actions satisfy~\eqref{eq:constraints} and the following stability certificate holds.
	\begin{definition}
		The equilibrium point $x=0$ of the closed-loop system~\eqref{eq:cl_sys} is locally exponentially stable with a decay rate $\alpha$ and a domain of attraction $\mathcal E_P$ if there exist scalars $C_0>0$ and $\alpha\in (0,1)$ such that
		$\|x_t\| \le C_0\alpha^{(t-t_0)}\|x_0\|$
		for any $x_0\in \mathcal E_P$.\frqed
	\end{definition}
	
	A standard result for testing local exponential stability of the equilibrium point adopted from~\cite{aitken1994exponential} is provided next.
	\begin{lemma}\label{lem:CtrlLem_2_constrained}
		Let $V:[0,\infty)\times \mathcal E_P\to \mathbb R$ be a continuously differentiable function such that the inequalities~\eqref{eq:CtrlLem2} hold for any $t\ge t_0$ and $x\in \mathcal E_P$ along the trajectories of the system~\eqref{eq:cl_sys} where $\gamma_1$, $\gamma_2$, and $\alpha$ are positive scalars. Then the equilibrium $x=0$ for the system~\eqref{eq:cl_sys} is locally exponentially stable with a decay rate $\alpha$ and a domain of attraction $\mathcal E_P$.\frqed
	\end{lemma}
	
	The following design theorem provides a method to construct a stabilizing policy such that the origin is a locally exponentially stable equilibrium of the closed-loop system and constraint satisfaction is guaranteed within a prescribed ellipsoid of attraction $\mathcal E_P\subset \mathcal X$ without knowing the nonlinearity~$\phi(\cdot)$.
	\begin{theorem}\label{thm:constrainedADP}
	Fix $\alpha\in (0,1)$ and $\liphat$. Suppose $\liphat\ge \lip$, and there exist matrices $S=S^\top\succ 0$, $Y$, and a scalar $\nu>0$ such that the LMI conditions~\eqref{eq:mainLMI_d} and
	\begin{equation}
	\label{eq:mainLMI_c_constrained}
	\begin{bmatrix}
	1 & \xi_i^\top Y \\
	\star & S
	\end{bmatrix} \succeq 0
	\end{equation}
	for every $i=1,\ldots, n_c$.
	Then, the equilibrium $x=0$ of the closed-loop system~\eqref{eq:cl_sys} is locally exponentially stable with a decay rate $\alpha$ and a domain of attraction $\mathcal E_P$ defined in~\eqref{eq:ellipsoid_P}. Furthermore, given that the initial state $x_0\in \mathcal E_P$, then the control actions $u_t$ satisfy the constraints~\eqref{eq:constraints} for all $t\ge 0$.
	\end{theorem}
	
	\begin{proof}
	From Theorem~\ref{thm:control_1}, we know that~\eqref{eq:CtrlLem2} holds. Taking Schur complements of~\eqref{eq:mainLMI_c_constrained} yields~\eqref{eq:ellipsoid_containment} which, by Lemma~\ref{lem:ellipsoid}, implies that $\mathcal E_P\subset \mathcal X$ and hence, the input constraints are satisfied for all $t\ge 0$ by the closed-loop system with policy $u=K_0x$, because $x_0\in\mathcal E_P$. Thus, all the conditions of Lemma~\ref{lem:CtrlLem_2_constrained} are satisfied, which concludes the proof. \frQED
	\end{proof}%
	\begin{remark}
	Note that the conditions~\eqref{eq:mainLMI_d} and~\eqref{eq:mainLMI_c_constrained} are LMIs in $S$, $Y$, and $\nu$ for a fixed $\liphat$. Therefore one can maximize the volume of $\mathcal E_P$ by solving a constrained convex program with cost function $-\log|S|$ (the log-determinant of $S$) subject to the constraints~\eqref{eq:mainLMI_d} and~\eqref{eq:mainLMI_c_constrained} while line searching for $\alpha$. This will reduce the conservativeness of the domain of attraction.
	\end{remark}

	\subsubsection{Safely initialized input-constrained PI}
	By adopting the work of \cite{abu2005nearly,lin2017optimal,zhang2009neural,modares2013adaptive,7302057} for input-constrained/actuator saturated ADP we choose a cost function of the form
	\begin{equation}
	\label{eq:cost_constrained}
	\mathcal U(x, u) = Q(x) + 2\int_0^u \left(\bar u \tanh^{-1}(\upsilon/\bar u)\right)^\top R \;\mathrm{d}\upsilon,
	\end{equation}
	where $Q(x):\mathbb R^{n_x}\to \mathbb R$ is a positive definite function satisfying $Q(0)=0$ and $R \succ 0$.  
	
	We begin by demonstrating that the constrained policy is an admissible policy on its domain of attraction.
	\begin{theorem}
	Let $\mathcal U$ be defined as in~\eqref{eq:cost_constrained}. Then the initial control policy $u_0=K_0x$ obtained by solving~\eqref{eq:mainLMI_d} and~\eqref{eq:mainLMI_c_constrained} is an admissible control policy on $\mathcal{E}_{P}$.
	\end{theorem}
	\begin{proof}
	By definition $Q(0)=0$. Also, the integrand in~\eqref{eq:cost_constrained} is zero when the upper limit is zero. Therefore, $\mathcal U(0,0)=0$.  For any $x_0\in\mathcal E_P$, $u_0$ is a stabilizing constrained control policy, therefore, $\|x_t\|\to 0$ and $\|u_t\|\to 0$ as $t\to \infty$. Hence, $\mathcal U\to 0$ as $t\to\infty$. The rest of the proof follows identically as in the proof of Theorem~\ref{thm:admissible_1}.\frQED
	\end{proof}	
	Since the control policy is constrained, we can initialize ADP safely using the neural approximator~\eqref{eq:Jk} as discussed in the previous subsection. The policy evaluation step is given by
	\begin{align}
	&\nonumber \omega_{k+1}^\top \big(\psi(x_t)- \gamma\psi(x_{t+1})\big) \\ 
	&\quad = Q(x_t) +  2\int_0^{u_k(x_t)} \left(\bar u \tanh^{-1}(\upsilon/\bar u)\right)^\top R\,\mathrm{d}\upsilon	
	\nonumber
	\\
	\label{eq:constrained_policy_evaluation_nn}
	&\quad = Q(x_t) + 2\bar u u^\top R\tanh^{-1}(u/\bar u) \\
	\nonumber	
	&\qquad\qquad\qquad + \bar u^2 \diag(R)^\top \begin{bmatrix}\ln(1-u_1^2/\bar u^2) \\ \ln(1-u_2^2/\bar u^2) \\ \vdots \\ \ln(1-u_{n_u}^2/\bar u^2)\end{bmatrix},\nonumber
	\end{align}
	where $u_{1},\ u_{2},\ u_{3},\ \cdots, u_{n_u}$  are the individual components of the vector $u$. Subsequently, the policy improvement step is given by
	\begin{equation}
	\label{eq:constrained_policy_improvement_nn}
	u_{k+1} = -\bar u\tanh\left[\frac{\gamma}{2\bar u}R^{-1}B^\top \nabla \psi (x_{t+1})^\top\omega_{k+1}\right],
	\end{equation}
	which satisfies the control constraints, since $\|\tanh(\cdot)\|_\infty \le 1$.
	
	Since the initial control policy is constrained and admissible, one can use~\cite[Theorem 2]{lin2017optimal} to prove convergence of the value function and the control policy to the optimal using the constrained policy iteration steps~\eqref{eq:constrained_policy_evaluation_nn} and~\eqref{eq:constrained_policy_improvement_nn}.
	
	\subsection{Remarks on on-policy VI and Q-learning}
	The value iteration algorithm (see Algorithm~\ref{algo:VI}) does not generally require an admissible control policy in order to converge optimally using data. Although this is true in off-policy implementations (that is, when the updated control policy is not used on-line), in on-policy implementations, a lack of stabilizing initial policies could result in unsafe transient behavior unless the underlying system is open-loop stable, leading to unsafe exploration during the initial data collection phase.
	\begin{algorithm}[!ht]
		\caption{Safely Initialized VI for discrete-time systems}
		\label{algo:VI}
		\small
		\begin{algorithmic}[1]
			\Require Termination condition constant $\epsilon_\mathrm{ac}$
			\Require Historical data $\mathcal D$
			\State Estimate Lipschitz constant $\liphat$ using Algorithm~\ref{algo:KLL}
			\Require Compute stabilizing control gain $K_0$ via SDP~\eqref{eq:mainLMI_d}
			\State Fix safe initial control policy $u_0(x)=K_0x$
			\While {$\left\|\mathcal J_k - \mathcal J_{k-1}\right\|\geq \epsilon_\mathrm{ac}$}
			\State Solve for the value $\mathcal J_{k}(x)$ using
			\[
			\mathcal J_{k+1}(x_t)=\mathcal U(x_t,u_k(x_t))+\gamma \mathcal J_{k}(x_{t+1}).
			\]
			\State Update the control policy $u_{(k+1)}(x)$ using
			\[
			u_{k+1}(x_t)=\argmin_{u(\cdot)}\big(\mathcal U(x_t,u_k(x_t))+\gamma \mathcal J_{k+1}(x_{t+1})\big).
			\]
			\State $k:=k+1$
			\EndWhile
		\end{algorithmic}
	\end{algorithm}
	
	Q-learning is a provably convergent direct optimal
	adaptive control algorithm and model-free reinforcement
	learning technique 
	\cite{Watkins1992,Tsitsiklis1994,mehta2009q,VAMVOUDAKIS201714}. Q-learning can
	be used to find an optimal action-selection policy based on measurements
	of previous state and action observations controlled
	using a sub-optimal policy. In most of the existing work the reward/cost function is manipulated to guarantee correction of the unsafe actions in the learning phase. Our proposed method does not require a corrective modification of the reward/cost function on-line for safety. Instead, historical data and solving SDPs based on Lipschitz estimation is used to generate safe control policies that enables safe data collection during on-policy Q-learning implementation, because the states are guaranteed not to diverge with the initial policy (this divergence could happen if the initial policy was unsafe).

	\section{Numerical Examples}\label{sec:ex}
	\subsection{Nonlinear torsional pendulum}
	We demonstrate our proposed approach using the torsional pendulum which is modeled by discretizing the system
	\begin{subequations}\label{eq:nonlinear_torsional_pendulum}
	\begin{align}
	\dot \theta &= \omega, \\
	J \dot \omega &= u - Mgl\sin\theta - f_d\omega,
	\end{align}
	\end{subequations}
	with mass $M=0.333$~Kg, length $l=0.667$~m, acceleration due to gravity $g=0.981$~m/s$^2$, friction factor $f_d=0.2$, and moment of inertia $J=0.1975$~Kg-m$^2$. With Euler discretization and a sampling time of $\tau=0.01$~s, we get a discrete-time model of the form~\eqref{eq:true_sys} with
	\[
	x = \begin{bmatrix} \theta \\ \omega \end{bmatrix},\;
	A = I + \tau \begin{bmatrix}
	0 & 1 \\ 0 & -f_d
	\end{bmatrix},\; B= \tau\begin{bmatrix}
	0 \\ 1
	\end{bmatrix}, \; G = \tau\begin{bmatrix}
	0 \\ -1
	\end{bmatrix}.
	\] We assume that the nonlinearity $\phi = Mgl\sin\theta/J$ is completely unknown; clearly $\phi(\cdot)$ has a Lipschitz constant $
	\lip = Mgl/J = 11.038$, which is also unknown to us.
	
	In the data collection phase, we initialize the system~\eqref{eq:nonlinear_torsional_pendulum} from ten different initial conditions in the space $[-\pi, \pi]\times [-2, 2]$ and collect data each $0.1$~s, leading to a total dataset of $N=50$~samples. Note that the initialization procedure of~\cite{lin2017optimal} requires 400 data points, which is considerably more than ours, and in that procedure, the original policy in the pre-training phase is not guaranteed to be admissible. Automatic relevance determination reveals that the nonlinearity only acts through the second state, and the argument of the nonlinearity is $q=\theta$. Proceeding as in Algorithm~\ref{algo:KLL}, we perform cross-validation using an Epanechnikov kernel with bandwidth $h_n=0.05$ and choose $\beta=0.01$. This yields the overestimate $\liphat=11.511>\lip$. Using this Lipschitz estimate, we solve~\eqref{eq:mainLMI_d} with $\alpha=0.95$ and $\nu=1$ for an initial value function $x^\top P x$ and control policy estimate $K_0x$.
	
	We construct a $2-11-1$ value function neural approximator with a set of polynomial basis functions
	\begin{align}
	\nonumber \psi(x_1,x_2) &= \left\{
	\frac{x_1^2}{2},
	\frac{x_2^2}{2},
	x_1x_2,
	\frac{x_1^2x_2}{2} ,
	\frac{x_1x_2^2}{2} ,
	\frac{x_1^4}{4}	,
	\frac{x_2^4}{4} ,\right.\\
	\label{eq:example1_basis}
	&\left. \hspace{2em}\!
	\frac{x_1^3}{3} ,
	\frac{x_2^3}{3} ,
	\frac{x_1^2x_2^2}{2} ,
	\frac{x_1^4x_2^4}{4} \right\},
	\end{align}
	where $x_1$, $x_2$ denote the first and second components of $x$, respectively. Our initial weight vector is set to $$\omega_0 = \begin{bmatrix}
	2P_{11} & 2P_{22} & P_{12}+P_{21}& 0 & \cdots & 0
	\end{bmatrix}^\top, $$ where $P_{ij}$ is the $(i,j)$th element of $P$. We fix the learning rate at $\eta=10^{-4}$ and the forgetting factor $\gamma=0.95$.
	
	\subsubsection*{Unconstrained scenario}
	\begin{figure}
		\centering
		\includegraphics[width=0.9\linewidth]{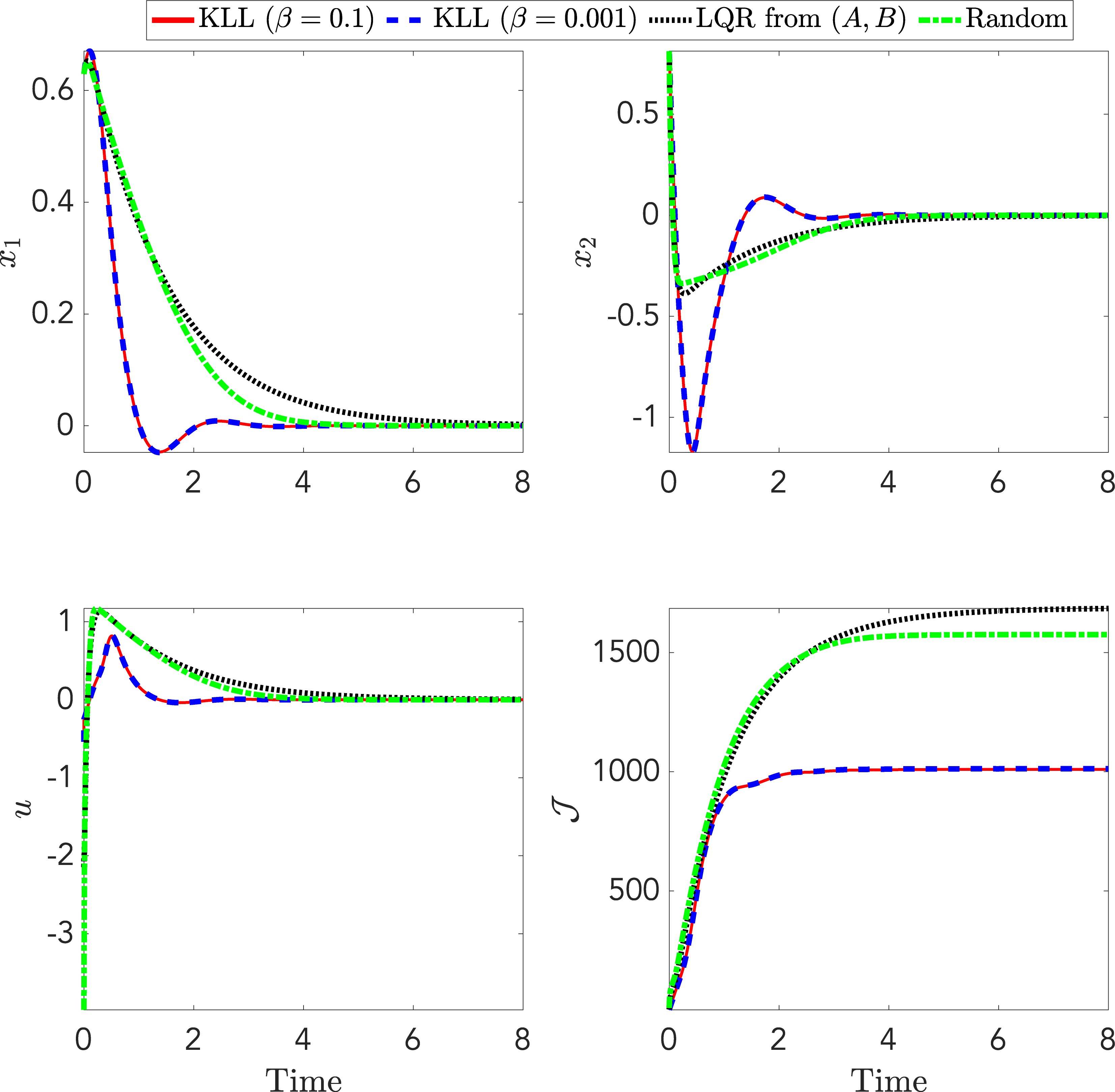}
		\caption{Comparison of states $x_t$, inputs $u_t$, and cost function values for unconstrained ADP with safe initialization for Lipschitz estimates with increasing confidence $\liphat(\beta=0.1)=11.14$ and $\liphat(\beta=0.001)=12.29$. We also compare our work to an LQR controller that is known to work for linear systems.}
		\label{fig:compareunconstrained}
	\end{figure}
	We first test the unconstrained scenario, where the cost function is $\sum \|Qx\|_1 + u^\top R u$, with $Q=I_2$ and $R=0.5$, and compare four initial policies and value functions obtained via: (i) kernelized Lipschitz learning with $\beta=0.1$; (ii) kernelized Lipschitz learning with $\beta=0.001$; (iii) solving an algebraic Riccati equation and ignoring the nonlinearity; and (iv) randomly initializing with small weights from a normal distribution with small variance and zero mean (which is by far the most common initializer). The comparison study results are shown in Figure~\ref{fig:compareunconstrained}. We observe that all of the methods (i)--(iv) listed above work, and result in stabilizing control policies that result in the state of the torsional pendulum to converge to its equilibrium. Interestingly, both the Lipschitz constant estimates result in similar trajectories implying that the SDPs~\eqref{eq:mainLMI_d} are not extremely sensitive to the Lipschitz estimate. However, based on the $\mathcal J$ subplot which shows the variation of $\sum x^\top Q x + u^\top R u$ with time, there is a slight improvement of performance in the $\beta=0.1$ (continuous red line) case compared to the $\beta=10^{-3}$ (dashed blue line) case since the Lipschitz estimate in the former is closer to the true Lipschitz constant. As expected, the cost incurred by the control policy ignoring the nonlinearity (dotted black line) is by far the worst, since the control actions required early on are of larger magnitude and the tracking performance is severely compromised. Randomly selecting weights also results in worse performance than our proposed method, as the cost incurred is increased due to oscillatory behaviour in the states and poor tracking in the initial time frame. Summarily, this experiment demonstrates the effectiveness of the proposed approach and its robustness to Lipschitz estimate conservatism.
	
	In Fig.~\ref{fig:plottorsional_VI}, we demonstrate the on-policy value iteration algorithm with safe initialization. All initial conditions converge to the origin using our proposed approach. In constrast, randomly initializing a policy and value as is typical in on-policy value iteration results in the states initially diverging (not shown in the plot) and poor performance before the rank condition is reached for determining a least-squares solution to update the neural weights.
		\begin{figure}
		\centering
		\includegraphics[width=.9\linewidth]{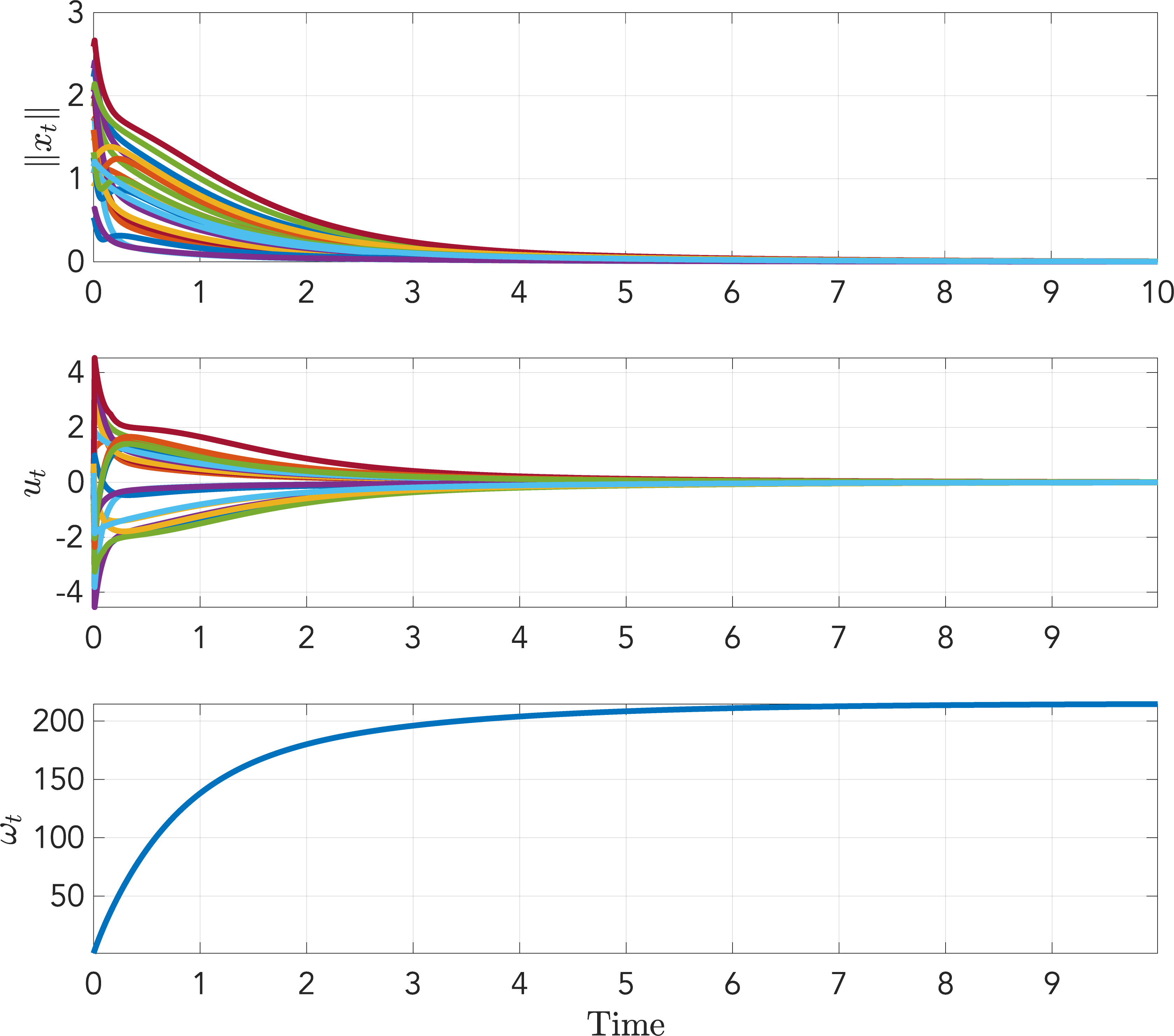}
		\caption{Illustration of constrained-input value-iteration based ADP with safe initialization. The top plot shows the variation of $\|x\|$ over time. The middle plot demonstrates that input constraints are satisfied for all $t$, and the bottom plot demonstrates that the initial control policy was close to the optimal, but learning was necessary to change the weights to the optimal values.}
		\label{fig:plottorsional_VI}
	\end{figure}

	\subsubsection*{Constrained scenario}
	\begin{figure}
		\centering
		\includegraphics[width=.9\linewidth]{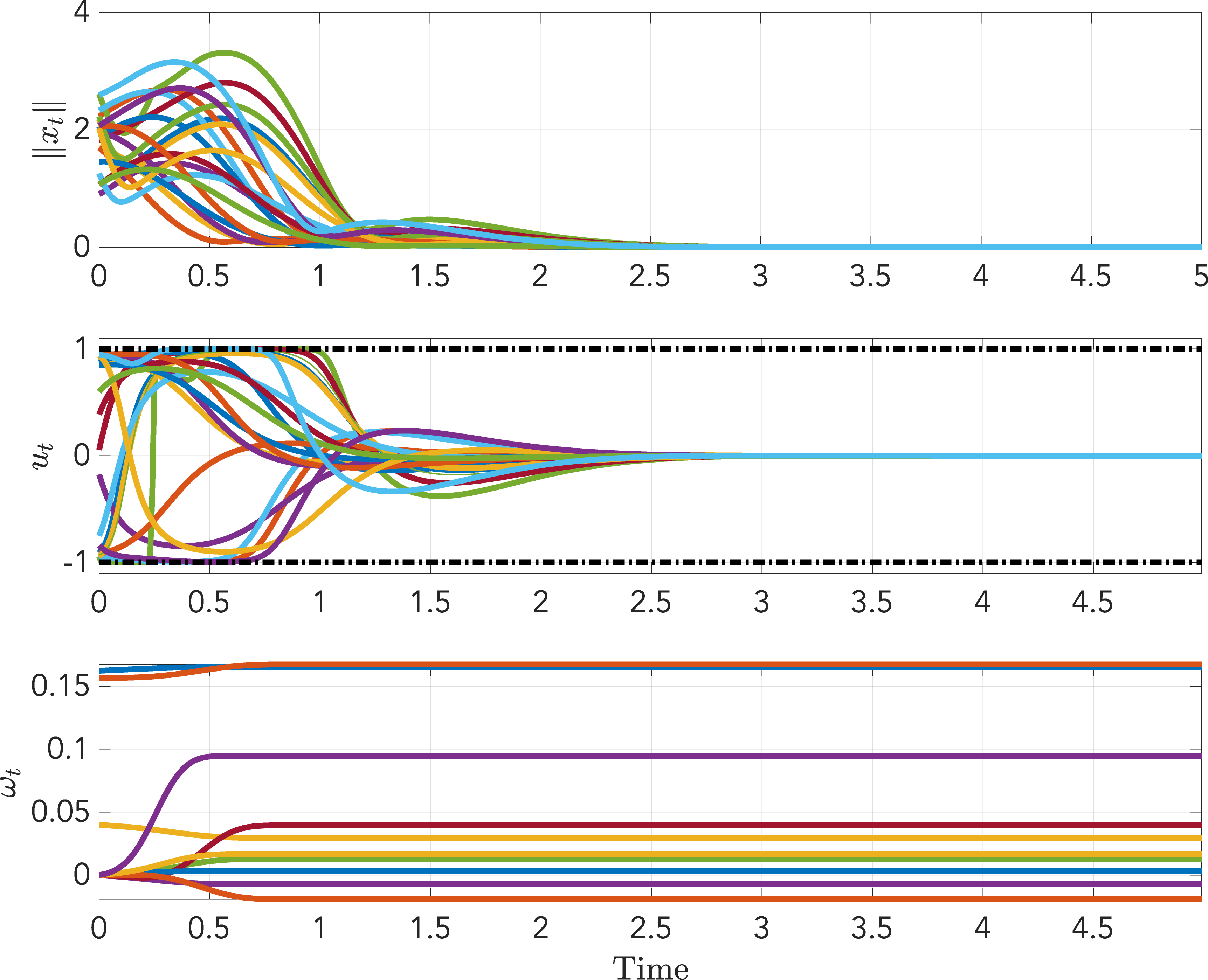}
		\caption{Illustration of constrained-input policy-iteration based ADP with safe initialization. The top plot shows the variation of $\|x\|$ over time. The middle plot demonstrates that input constraints are satisfied for all $t$, and the bottom plot demonstrates that the initial control policy was close to the optimal, but learning was necessary to change the weights to the optimal values.}
		\label{fig:plottorsional}
	\end{figure}
	We also test the scenario where the control actions are constrained by $|u|\le 1$. In this case, we use the cost functional defined in~\eqref{eq:cost_constrained} with $\bar u=1$, $Q=I_2$ and $R=0.5$. We begin by solving~\eqref{eq:mainLMI_d} and~\eqref{eq:mainLMI_c_constrained} with $\nu=1$ and $\alpha=0.95$ to get $P$ and $K_0$, as in the previous subsection. We also select the same basis functions~\eqref{eq:example1_basis}. We randomly initialize (using 20 random initial conditions) the system~\eqref{eq:nonlinear_torsional_pendulum} from within the domain of attraction of the initial control policy, that is, from within the set $\{x^\top P x \le 1\}$. We know from Theorem~\ref{thm:constrainedADP} that this ensures that the initial control policy will satisfy input constraints. Consequently, because the policy improvement step is also guaranteeed to satisfy input constraints and the initial policy is stabilizing, there are no constraint violations, and the initialization is deemed safe. The performance of the proposed algorithm is provided in Figure~\ref{fig:plottorsional}. The convergence of $\|x_t\|$ to zero and the satisfaction of input bounds are illustrated. Finally, we demonstrate the convergence of the neural weights $\omega_t$, noting that learning did occur, that is, the weights were not static throughout the simulation (which would indicate that the initial policy was optimal).

	\subsection{Large randomized linear system}
	In order to study the scalability of the proposed approach on higher state-space dimensions, we randomly generate a 20-state, 10-input linear system of the form $x_{t+1}= Ax_t + Bu_t$, where $B$ is known, and $A$ is unknown. We randomly choose $A_0\neq A$ to be a known matrix such that $(A_0,B)$ is controllable. Clearly, the unmodeled component is $\phi(x) = (A - A_0)x$. We assume we do not know $G$ and $C_q$, so both are set to identity matrices of appropriate dimensions. The initial dataset is generated using a small random perturbation signal on the unknown system and 500 data points are stored, from which we compute a Lipschitz constant for $\phi$ as $\liphat = 0.43$ (the true $\lip=0.34$) and a safe initial policy with $\beta=0.001$; this is illustrated in Figure~\ref{fig:figlargelinear}'s top left subplot; the shaded blue area is the subgraph of $\hat L_n$ within the $\beta$-confident support. The unknown system is simulated from 100 randomly selected initial conditions in $\mathbb R^{20}$ following a normal distribution with variance 2 and zero-mean. The cost function has the form~\eqref{eq:LQR_cost} with $Q=5I_{20}$, $R=2I_{10}$, and $\gamma=0.95$. In all cases, the weights of the $20-210-10$ neural approximator with learning rate $\eta=0.1$ converges within a few seconds, and the system stabilizes, in spite of $A$ being unstable, as shown in~Figure~\ref{fig:figlargelinear}.
	
	\begin{figure}[!ht]
		\centering
		\includegraphics[width=.9\linewidth]{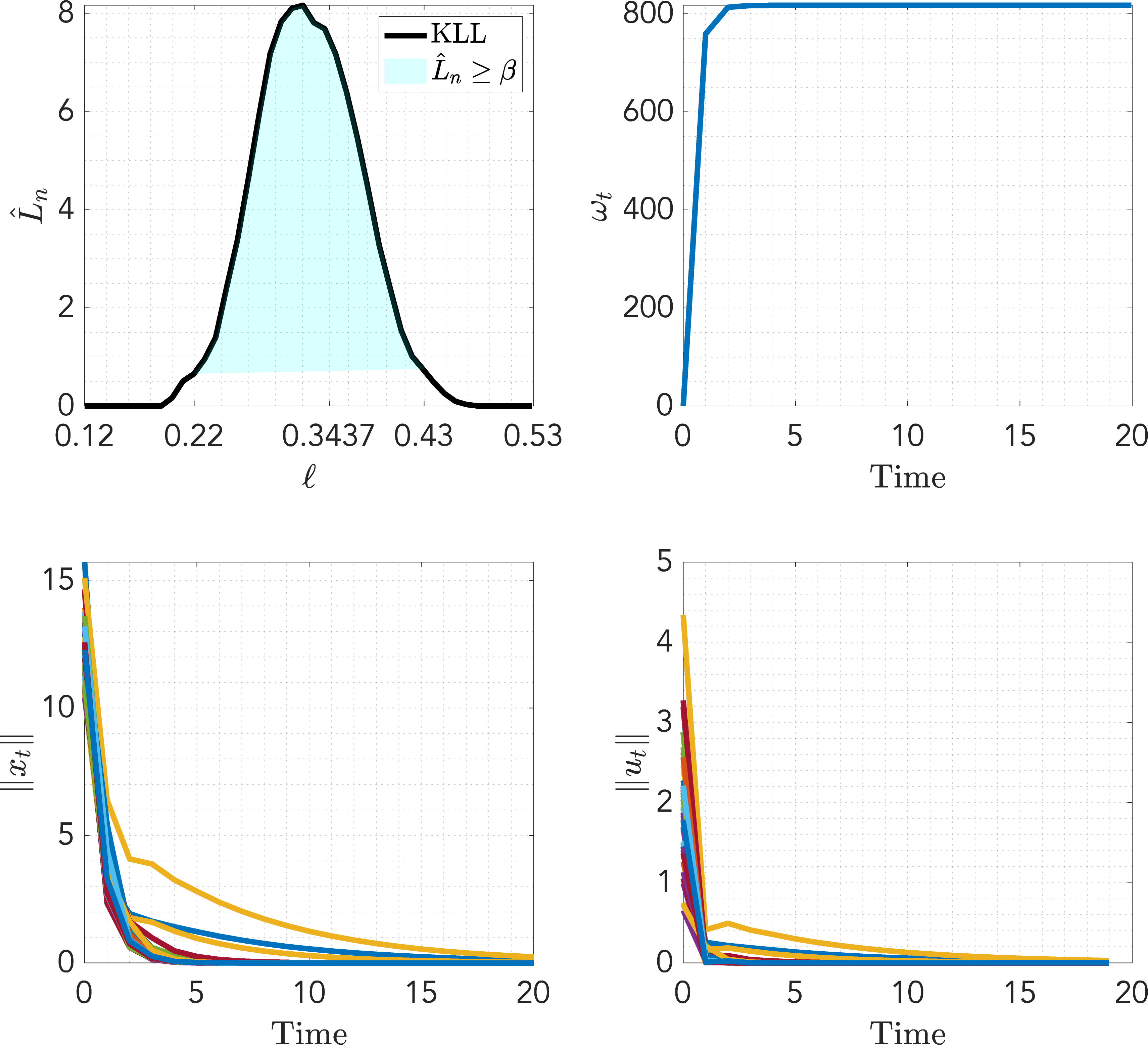}
		\caption{Kernelized Lipschitz estimate, and evolution of states, inputs, and weights for the unknown 20-state, 10-input linear system using on-policy policy iteration.}
		\label{fig:figlargelinear}
	\end{figure}

	\section{Conclusions}\label{sec:conc}
	This work provides a methodology for constructing admissible initial control policies for ADP methods using Lipschitz learning by using kernel density estimation and semi-definite programming. Such admissible controllers enable safe initialization, that is, with constraint satisfaction using only historical data, which is necessary not only in policy iteration methods, but also in value iteration and Q-learning for safely obtaining initial data on-line for on-policy learning when the underlying system is not open-loop stable. Simulations on a discretized torsional pendulum model and a high-dimensional linear system are provided to show the efficiency of our approach. Future research efforts will focus on more general costs and uncertain nonlinear safety constraints while ensuring feasibility with a high probability  in terms of regret.
		 
	\section*{Acknowledgments}
	We would like to thank Drs. Mouhacine Benosman and Piyush Grover at Mitsubishi Electric Research Laboratories, Cambridge, MA, USA, for their time and helpful insights. Kyriakos Vamvoudakis was supported in part by NSF under grant Nos. CPS-1851588 and S\&AS-1849198.

	\balance
	\setstretch{0.98}
	\bibliographystyle{IEEEtran}
	\bibliography{IEEEabrv,\jobname}%
	
\end{document}